\documentclass[11pt]{article}

\usepackage{amssymb,amsmath,amsthm}
\usepackage{graphicx, color}
\usepackage[arrow, matrix, curve]{xy}
\usepackage{txfonts}
\usepackage[top=1in, bottom=1in, left=1in, right=1in]{geometry}

\theoremstyle{definition}
\newtheorem{definition}{Definition}[section]

\theoremstyle{plain}
\newtheorem{theorem}[definition]{Theorem}

\newtheorem{lemma}[definition]{Lemma}

\newtheorem{proposition}[definition]{Proposition}
\newtheorem{remark}[definition]{Remark}
\newtheorem{assumption}[definition]{Assumption}

\theoremstyle{dotless}

\newcommand{\vega}{{\mbox{\Large \sf $\mathcal{\nu}$}}}

\DeclareMathOperator{\essinf}{ess \, inf\,}

\newcommand{\E}{\mathbb{E}}

\newcommand{\R}{\mathbb{R}}

\renewcommand{\phi}{\varphi}

\newcommand{\calf}{\mathcal{F}}
\newcommand{\lop}{{\mathcal{L}}}
\newcommand{\half}{\frac12}
\newcommand{\pa}{\partial}
\newcommand{\vol}{volatility }

\title{From Smile Asymptotics to Market Risk Measures}
\author{Ronnie Sircar\thanks{ORFE Department,
    Princeton University, Sherrerd Hall, Princeton NJ 08544; {\it
      sircar@princeton.edu}. Work partially supported by NSF grant
    DMS-0807440.} \and Stephan Sturm\thanks{Department of Mathematical Sciences, Worcester Polytechnic Institute,
    Stratton Hall, Worcester, MA 01609; {\it
      ssturm@wpi.edu}. Work partially supported by NSF grant
    DMS-0739195. This work was completed while Postdoctoral Research Associate at Princeton University.} }

\begin{document}
\maketitle
\begin{abstract} 
The left tail of the implied volatility skew, coming from quotes on out-of-the-money put options, can be thought to reflect the market's assessment of the risk of a huge drop in stock prices. 
We analyze how this market information can be integrated into the theoretical framework of convex monetary measures of risk. In particular, we make use of indifference pricing by dynamic convex risk measures, which are given as solutions of backward stochastic differential equations (BSDEs), to establish a link between these two approaches to risk measurement. We derive a characterization of the implied volatility in terms of the solution of a nonlinear PDE and provide a small time-to-maturity expansion and numerical solutions. This procedure allows to choose convex risk measures in a conveniently parametrized class, distorted entropic dynamic risk measures, which we introduce here, such that the asymptotic volatility skew under indifference pricing can be matched with the market skew. We demonstrate this in a calibration exercise to market implied volatility data.
\end{abstract}




\section{Introduction}
Risk measurement essentially conveys information
about tails of distributions. However, that information is also
contained in market prices of insurance securities that are
contingent on a large (highly unlikely) downside, if
we concede that those prices are mostly reflective of protection
buyers' risk aversion. Examples are out-of-the-money put options
that provide protection on large stock price drops, or senior
tranches of CDOs that protect against the default risk of say $15-30\%$
of investment grade US companies over a 5 year period.

A central regulatory and internal requirement in recent years, in
the wake of a number of financial disasters and corporate
scandals, has been that firms report a measure of the risk of
their financial positions. The industry-standard risk measure,
value-at-risk, is widely criticized for not being convex and
thereby possibly penalizing diversification, and a number of natural
problems arise:
\pagebreak
\begin{enumerate}
\item How to construct risk measures with good properties.
\item Computation of these risk measures under typical financial
models.
\item Choice: what is the ``right'' risk measure?
\end{enumerate}
The first issue has been extensively studied in the static case
\cite{Artzner,FollSchiedpaper} and recent developments in extending to dynamic
risk measures with good time-consistency and/or recursive properties
are discussed, e.g., in \cite{BK,KS2,muszar-long,FScB}.  However, concrete examples
of dynamic, time-consistent convex risk measures are scarce, and they
typically have to be defined abstractly, for example via the driver of
a backward stochastic differential equation (BSDE) or as the limit of
discrete time-consistent risk measures \cite{Stad}.  As a result,
intuition is lost, and there is at present little understanding what
the choice of driver says about the measure of risk. Or, to put it
another way, how can the driver be constructed to be consistent with
risk aversion reflected in the market?

Let $\xi$ be a bounded random variable representing a financial
payoff whose value is known at some future time $T<\infty$. A
classical example of a convex risk measure, the entropic risk measure, is related to exponential
utility: 
\begin{equation} \varrho(\xi) = \frac{1}{\gamma}\log\left(\E\left[e^{-\gamma
\xi}\right]\right) , \label{exputility}
\end{equation}
where $\gamma>0$ is a risk-aversion coefficient.
When extending to {\it dynamic} risk measures $\varrho_t(\cdot)$ adapted to some filtration $(\calf_t)$, a
desirable property is (strong) time-consistency 
\[ \varrho_s(-\varrho_{t}(\xi)) = \varrho_s(\xi), \qquad 0\leq s\leq t\leq T. \]
This flow property is important if $\varrho_t$ is used as a basis for a
pricing system. The static entropic risk measure (\ref{exputility})
generalizes to 
\begin{equation}
  \varrho_t(\xi) = \frac{1}{\gamma}\log\left(\E\left[e^{-\gamma
\xi}\mid \calf_t\right]\right). \label{entrodyn}
\end{equation}
The flow property follows simply
from the tower property of conditional expectations. However, finding
other directly-defined examples is not easy, and to have a reasonable
class of choices, we need to resort to more abstract constructions. 

In a Brownian-based model, time-consistent dynamic risk measures can
be built through BSDEs, as shown in \cite{BK,KS2} (compare also \cite{RoGi}), extending the work
of Peng \cite{peng}. That is, on a probability space with a
$d$-dimensional Brownian motion $W$ that generates a filtration
$(\calf_t)$, the risk measure of the $\calf_T$-measurable random
variable $\xi$ (taking values in $\R$ for simplicity) is computed from the solution
$(R_t,Z_t)$, which takes values in $\R\times\R^d$, of the 
BSDE
\begin{eqnarray*}
-dR_t & = & g(t,Z_t)\,dt - Z_t^*\,dW_t\\
R_T & = & -\xi,
\end{eqnarray*}
where $*$ denotes transpose.  Here the driver $g$, which defines the
risk measure, is Lipschitz and convex in $z$ and satisfies
$g(t,0)=0$. The solution is a process $R$, taking values in $\R$ that matches the {\it
  terminal} condition $-\xi$ on date $T$ (when $\xi$ is revealed and the
risk is known), and a process $Z$ taking values in $\R^d$ that, roughly speaking, keeps
the solution non-anticipating.  Then $\varrho_t(\xi):=R_t$ defines a
time-consistent dynamic convex risk measure.  However, the possibility
to offset risk by dynamically hedging in the market needs to be
accounted for. Setting aside technicalities for the moment, this
operation leads to a modification of the driver.

The left tail of the implied volatility skew observed in equity
markets is a reflection of the premium charged for out-of-the-money
put options. The bulk of the skew reveals the heavy left tail in the
risk-neutral density of the stock price $S_T$ at expiration, but the
very far left tail, where investor sentiment and crash-o-phobia takes
over, could be interpreted as revealing information about the
representative market risk measure and its driver $g$, if we assume
prices are consistent with this kind of pricing mechanism. The
question then is to extract constraints on the driver from the
observed tails of the skew, an inverse problem. In the application to equity options in Section \ref{ec}, we assume the mid-market option prices reflect the premium a risk-averse buyer is willing to pay. We do not relate buyers' and sellers' prices to the bid-ask spread, since that is more likely related to the market maker's profit. 

To put our analysis into a broader framework, we observe that the underlying structural question is the inference of preference structures from observable data. The idea of using (at least in theory)  observable consumption and investment streams to reveal the preference structure of a rational utility maximizing investor dates back to Samuelson in the 1940s and  Black in the 1960s - for a recent overview on this ``backward approach'' to utility theory we refer to \cite{cho}. The spirit of our presentation is a similar one, except we deal with dynamic risk measures rather than utility functions, and the observable data are not given as consumption and investment strategies but as readily available market implied volatilities.

The main goal of the current article is to develop short-time asymptotics that can be used for the
inverse problem of extracting information about the
driver $g$ from the observed skew.  This could be used to construct an approximation to the driver and then to value more exotic derivatives in a way consistent with the risk measure. The information could also be used to quantify market perception of tail events, particularly when they depart from usual. Some studies have discussed the steepening of skew slopes in the run up to financial crises without a corresponding overall raise in volatility level. Inferring, fully or partially, a risk measure driver, could be used for detection of increased wariness of a crash.

Berestycki {\it et al.} \cite{BBF} presented short-time
asymptotics for implied volatilities for {\it no arbitrage pricing}
under a given risk-neutral measure in stochastic volatility models. Further work in this direction
includes, among others, \cite{FFF,FoJa09a,JFL} and references therein.
In Section \ref{Sec2}, we extend this analysis to
the nonlinear PDEs characterizing indifference pricing under dynamic
convex risk measures. 

We find (Theorem \ref{maintheorem}) that the zero-order term in the
short-time approximation is the same as found in no-arbitrage pricing
by \cite{BBF}. The next order term is the solution of an inhomogeneous
linear transport equation that sees only a particular slope of the
partially Legendre-transformed driver, but is independent of the size
of the options position (see equation \eqref{firstPDE}).

Section \ref{ec} illustrates the theoretical findings by focusing on a particular class of drivers, introducing distorted entropic convex dynamic risk measures. First we develop explicit calculations for the small time expansion in the 
Hull-White stochastic volatility model to illustrate the impact of the distortion parameter
on the implied volatility skew. In Section \ref{calib}, as a proof of concept, we perform a preliminary calibration exercise of the short time asymptotics to S\&P 500 implied volatilities close to maturity. This allows to estimate the stochastic volatility model parameters from the liquid central part of the skew, and to recover the distortion parameter of our family of dynamic convex risk measures from far out-of-the-money put options.

We illustrate the parameter impact  for longer dated options in a numerical study (via the pricing PDE) of arctangent stochastic volatility driven by an Ornstein-Uhlenbeck process.  Section \ref{conc} contains the conclusions and Section \ref{proofs} gives the more technical proofs omitted in the exposition.

\section{Heuristics and Statement of Results\label{Sec2}}

We consider a model of a financial market consisting of a risk-free
bond bearing no interest and some stock following the stochastic
volatility model on the filtered probability space $(\Omega,
\mathcal{F}, (\mathcal{F}_t),P)$ 
\begin{equation}\label{SVM}
\left\{
\begin{array}{lll}
dS_t &= \mu(Y_t) S_t \,dt + \sigma(Y_t) S_t \, dW_t^1, \qquad & S_0 = S; \\
dY_t & = m(Y_t) \,dt + a(Y_t) \bigl(\rho \, dW_t^1 +\rho'\,
dW_t^2\bigr), \qquad & Y_0 = y, 
\end{array}
\right. 
\end{equation}
where $W^1$, $W^2$ are two independent Brownian motions generating
$(\mathcal{F}_t)$ and $\rho' = \sqrt{1-\rho^2}$. 

\begin{assumption}\label{mainassump}
We assume that:
\begin{itemize}
\item[i)] $\sigma$, $a \in C_{loc}^{1+\beta}(\mathbb{R})$, where 
$C_{loc}^{1+\beta}(\mathbb{R})$ is the space of differentiable
functions with locally H\"{o}lder-continuous derivatives with
H\"{o}lder-exponent $\beta>0$; 

\item[ii)] both $\sigma$ and $a$ are bounded and bounded away from
zero: 
\[ 0<\underline{\sigma}<\sigma <\overline{\sigma}<\infty, \quad
\mbox{and} \quad 
0<\underline{a}<a <\overline{a}<\infty;\]

\item[iii)] $\mu$, $m \in
C_{loc}^{0+\beta}(\mathbb{R})$,  and $\vert \mu \vert < \overline{\mu} <
\infty$. 
\end{itemize}
\end{assumption}

The pricing will be done via the indifference pricing mechanism
for dynamic convex risk measures, which are introduced in the next
subsection.

\subsection{Dynamic Convex Risk Measures, Indifference Pricing and BSDEs}

\begin{definition}
We call the family $\varrho_{t} : L^\infty(\Omega, \mathcal{F}_T, P) \to
L^\infty(\Omega, \mathcal{F}_t, P)$, $0\leq t \leq T$, a convex dynamic
risk measure, if it satisfies for all $t \in [0,T]$ and all $\xi$,
$\xi^1$, $\xi^2 \in L^\infty(\Omega, \mathcal{F}_T,P)$ the following
properties. 
\begin{equation}\label{charrho}
\begin{array}{rl}
i) & \mbox{Monotonicity: }\xi^1 \geq \xi^2\,P\mbox{-a.s implies
}\varrho_{t}(\xi^1) \leq \varrho_{t}(\xi^2);\\ 
ii) & \mbox{Cash invariance: }\varrho_{t}(\xi+m_t) = \varrho_{t}(\xi)-m_t
\mbox{ for all }m_t \in L^\infty(\Omega, \mathcal{F}_t,P);\\ 
iii) & \mbox{Convexity: } \varrho_{t}(\alpha \xi^1 + (1-\alpha) \xi^2)
\leq \alpha \varrho_{t}(\xi^1) + (1-\alpha)\varrho_{t}(\xi^2) \mbox{ for all
} \alpha \in [0,1];\\ 
iv) & \mbox{Time-consistency } \varrho_{t}(\xi^1) = \varrho_{t}(\xi^2)
\mbox{ implies }  \varrho_{s}(\xi^1) = \varrho_{s}(\xi^2) \mbox{ for all }
0\leq s\leq t. 
\end{array}\nonumber
\end{equation}
\end{definition}
We note that if the risk measure is additionally normalized,
i.e. $\varrho_{t}(0)=0$ for all $t\in [0,T]$, then iv) is equivalent
to the stronger property $\varrho_{s}(-\varrho_{t}(\xi)) =
\varrho_{s}(\xi)$ for all $0\leq s\leq t$ \cite[Lemma 3.5]{KS2}. The
risk measure $\varrho_{t}(\xi)$ should be understood as the risk
associated with the position $\xi$ at time $t$.

If $\varrho_{t}$ is normalized, this is nothing else than the minimal
capital requirement at time $t$ to make the position riskless since
$\varrho_{t}( \xi + \varrho_{t}(\xi)) =0$. In this static setting, the
certainty equivalent price of a buyer of a derivative $\xi \in
L^\infty(\Omega,\mathcal{F}_T,P)$ at time $t$ is just the cash amount
for which buying the derivative has equal risk to not buying it.

In fact we are much more interested in the case where the buyer of the
security is allowed to trade in the stock market to hedge her risk. In
describing admissible strategies we follow the setting of continuous
time arbitrage theory in the spirit of Delbaen-Schachermayer (for an
overview, we refer to the monograph \cite{DS}). Denote therefore by
$\Theta_t$ the set of all admissible hedging strategies from time $t$
onwards, i.e. all progressive processes such that $\theta_t = 0$ and $
\int_t^u \theta_s (\mu(Y_s) \, ds + \sigma(Y_s) \, dW^1_s)$ exists for
all $u \in ]t,T]$ and is uniformly bounded from below, and set
\begin{equation*}
\mathcal{K}_t := \left\{ \int_t^T \theta_s (\mu(Y_s) \, ds +
\sigma(Y_s) \, dW^1_s) \, : \, \theta \in \Theta_t\right\}.
\end{equation*}
The set of all superhedgeable payoffs is then given by
$\mathcal{C}_t := (\mathcal{K}_t - L^0_+) \cap L^\infty$,
where $L^0_+$ denotes the set of all almost surely non-negative random
variables.

The residual risk at time $t$ of the derivative $\xi\in L^\infty(\Omega,
\mathcal{F}_T,P)$ after hedging is given by 
\begin{equation}\label{posthedge}
\hat{\varrho}_{t}(\xi) := \underset{h \in \mathcal{C}_t}{\essinf} \varrho_{t}(\xi+h).
\end{equation}
Thus, assuming that the buyer's wealth at time $t$ is $x$, her
dynamic indifference price $P_t$, which can be viewed as the certainty
equivalent after optimal hedging in the underlying market, is given via 
$\hat{\varrho}_{t}(x+\xi-P_t) = \hat{\varrho}_{t}(x)$,
whence, using cash invariance,
\begin{equation}
P_t = \hat{\varrho}_{t}(0) -\hat{\varrho}_{t}(\xi).\label{indiffdef}
\end{equation}
We note, while restricting ourselves to the buyer's indifference
price, all our considerations are easily adaptable to the seller's
indifference price by a simple change of signs of $\xi$ and $P_t$ in
\eqref{indiffdef}.

A convenient class of dynamic convex risk measures to which we will
stick throughout this paper are defined from solutions of BSDEs. Assume that $g:
[0,T] \times \Omega \times \mathbb{R}^2 \to \mathbb{R}$ is a $\mathcal{P} \otimes
\mathcal{B}(\mathbb{R}^2)$ predictable function which is continuous,
convex, and quadratic (i.e. bounded in modulus by a quadratic
function) in the $\mathbb{R}^2$-component.
Next, let $\xi \in
L^{\infty}(\Omega, \mathcal{F}_T,P)$ be a given bounded financial
position. Then the BSDE
\begin{align}\label{bsde}
R_t = -\xi + \int_t^T g(t, \omega, Z^1_s,Z^2_s) \, ds - \int_t^T Z^1_s\, dW_s^1 -
\int_t^T Z^2_s\, dW_s^2 
\end{align}
admits a unique $\mathcal{F}_t$-adapted solution $(R_t,Z^1_t,Z^2_t)$,
which defines a dynamic convex risk measure via $\varrho_{t}(X) := R_t$
\cite[Theorem 3.21]{BK}.

The existence of a solution of the BSDE \eqref{bsde} in this quadratic
setting was first proved by Kobylanski \cite[Theorem 2.3]{Kob}, with
some corrections to their arguments given by Nutz \cite[Theorem 3.6]{Nut},
whereas the uniqueness follows from the convexity of the driver as
shown in \cite[ Corollary 6]{BH}. From a financial perspective, the
components $Z^1$, $Z^2$ of the ``auxiliary'' process $Z$ can be
interpreted as risk sources, describing the risk stemming from the
traded asset and the volatility process respectively.

\subsection{Transformed BSDE under Hedging}
To assure the solvability of the BSDEs and PDEs that arise in our
setting, we have to restrict the class of admissible drivers.
Throughout, subscripts of functions indicate in the PDE context partial
derivatives with the respect to the respective components. 
\begin{definition}\label{admiss}
We call a  function $g : \mathbb{R}^2 \to \mathbb{R}$ a strictly quadratic driver  (normalized strictly quadratic driver)  if it satisfies the
following 
conditions i)-iii) (resp. o)-iii)): 
 \begin{equation}\label{charg}
\begin{array}{rl}
o)& g(0,0) =0;\\
i) & g\in C^{2,1}(\mathbb{R}^2);\\
ii) & g_{z_1 z_1}(z_1,z_2)>0 \quad \mbox{for all } (z_1,z_2) \in \mathbb{R}^2;\\
iii) & \mbox{there exist constants } c_1, \,c_2>0 \mbox{ such that}\\ 
&\quad c_1 \Bigl(\frac{z_1^2}{4c_1^2} -(1+z_2^2)\Bigr) \leq
g(z_1,z_2)  \leq c_2 \bigl(1+ z_1^2 +z_2^2\bigr) \quad
\mbox{for all } (z_1,z_2) \in \mathbb{R}^2. 
\end{array}\nonumber
\end{equation}
\end{definition}
The normalization of the driver (condition o)) corresponds to the normalization of
the risk measure.  
\begin{remark}
  To ease the presentation, we only work with drivers that do not depend
  explicitly on time or on $\omega$. While any dependence on $R_t$ would destroy the cash invariance, it is not difficult to add an additional dependence
  of $g$ on time (requiring some local H\"{o}lder-continuity in the time component of $g$).  The small time expansion Section
  \ref{expansions} up to first order works in this case exactly as in the time independent case (evaluating the driver at $t=T$), and the higher order expansions can be adapted straightforwardly. For analysis of filtration-consistent, translation invariant nonlinear expectations, we refer to \cite{Coq}.
\end{remark}

In passing from the principal risk measure defined by $g$ to the
residual risk measure after hedging, as in \eqref{posthedge}, we will
need the Fenchel-Legendre transform of $g$ in its first component, namely
\begin{equation}
\hat{g}(\zeta,z_2) := \sup_{z_1 \in \mathbb{R}}\Bigl(\zeta z_1 -
g(z_1,z_2)\Bigr), \qquad \zeta\in\mathbb{R}. \label{ghatdef}
\end{equation}

\begin{lemma}\label{admissghat}
  Given that $g$ is a (normalized) strictly quadratic driver, then the
  risk-adjusted driver $\hat{g}$ defined in \eqref{ghatdef} 
is also a (normalized) strictly quadratic  driver.
\end{lemma}
\begin{proof}
To show $\hat{g}$ satisfies condition iii) of Definition \ref{admiss},
we fix $z_2$ and treat the function as classical Fenchel-Legendre
transform in one variable. Therefore it holds for proper, continuous
convex functions $f$, $g$, that $f \leq g$ implies $\hat{f} \geq
\hat{g}$ and $\hat{\hat{f}}=f$, \cite[Proposition E.1.3.1 and
Corollary E.1.3.6]{HUL}. So the statement is proved by noting that
\begin{equation*}
\sup_{z_1}\Bigl(\xi z_1 - c \bigl(1+ z_1^2 +z_2^2\bigr)\Bigr) = c
\Bigl(\frac{z_1^2}{4c^2} -(1+z_2^2)\Bigr)  
\end{equation*}
for any positive constant $c$.

To show $i)$ and $ii)$ in Definition \ref{admiss}, we note that
condition $iii)$ implies that $g$ is $1$-coercive in $z_1$,
i.e. $g(z_1,z_2)/\vert z_1\vert \to \infty$ as $z_1 \to \pm \infty$
for fixed $z_2$. Now we can use the fact, that the Fenchel-Legendre
transform of any $1$-coercive, twice differentiable function with
positive second derivative is itself $1$-coercive and twice
differentiable with positive second derivative, cf. \cite[Corollary
X.4.2.10]{HUL2}. Thus it remains only to prove the differentiability
of $\hat{g}$ with respect to $z_2$ which is a consequence of the
differentiability properties of $g$: writing down the difference
quotient and noting that the maximizer is differentiable, the positive
second derivative with respect to the first component yields the
existence of a finite limit. Finally $\hat{g}(0,0)=0$ follows from the
definition if $g(0,0)=0$.
\end{proof}

In other words, the class of strictly quadratic drivers is invariant under the
convex conjugation in the first component and the class of normalized
strictly quadratic drivers is an invariant subclass thereof. 

\begin{remark}\label{rem26}
We note that it is important in our setting to stick to the theory of
quadratic drivers, since if $g$ would be a Lipschitz driver, $\hat{g}$
would be no more a proper function. This fact is easily seen, since
from the Lipschitz condition it follows that 
\begin{equation*}
g(z_1,z_2) \leq L\Bigl( 1+ \sqrt{z_1^2+z_2^2}\Bigr) \leq \sqrt{2} L
\bigl(1 + \vert z_1\vert + \vert z_2 \vert \bigr)  
\end{equation*}
for some constant $L$ and hence
\begin{align*}
\hat{g}(\zeta ,z_2) = \sup_{z_1}\Bigl(\zeta z_1 - g(z_1,z_2)\Bigr)
&\geq \sup_{z_1}\Bigl(\zeta z_1 -  \sqrt{2} L \bigl(1 + \vert z_1\vert
+ \vert z_2 \vert \bigr)\Bigr) \\ 
&= \left\{ \begin{array}{rl} \infty & \mbox{if } \vert \zeta \vert >
    \sqrt{2}L \\  -\sqrt{2}L(1+ \vert z_2\vert)& \mbox{if } \vert
    \zeta \vert \leq \sqrt{2}L.\end{array}\right. 
\end{align*}
\end{remark}
From now on we will assume that $g$ is convex as a function on
$\mathbb{R}^2$ and an admissible driver. Our next step is to describe
the dynamic hedging risk in terms of BSDEs. These results are in
essence due to Toussaint, \cite[Section 4.4.1]{Tou}. Since his thesis
is not easily available, we will nevertheless
state the proofs here. It is convenient to introduce a notation for
the Sharpe ratio:
\begin{equation}
\lambda(y) := \frac{\mu(y)}{\sigma(y)}. \label{Sharpedef}
\end{equation}

\begin{proposition}\label{riskBSDE}
The risk of the financial position $\xi \in L^{\infty}(\Omega,
\mathcal{F}_T,P)$ under hedging is $\hat{\varrho}_{t}(\xi) = \hat{R_t}^{(\xi)}$
where $\hat{R}_t^{(\xi)}$ is given via the unique solution of the BSDE 
\begin{align}\label{transformBSDE}
\hat{R}_t^{(\xi)} = -\xi - \int_t^T \hat{Z}_s^1 \lambda(Y_s)
+\hat{g}\bigl(-\lambda(Y_s), \hat{Z}_s^2\bigr) \, ds -
\int_t^T \hat{Z}^1_s\, dW_s^1 - \int_t^T \hat{Z}^2_s\, dW_s^2. 
\end{align}
Moreover, $\hat{\varrho}_t$ is itself a dynamic convex risk measure.
\end{proposition}
\begin{proof}
It follows from the work of Kl\"{o}ppel and Schweizer \cite[Theorem
7.17]{KS2} that the risk is given via the BSDE 
\begin{equation*}
\hat{R}_t^{(\xi)} = -\xi + \int_t^T  \tilde{g}(Z_s^1,Z_s^2) \, ds - \int_t^T
\hat{Z}^1_s\, dW_s^1 - \int_t^T \hat{Z}^2_s\, dW_s^2,
\end{equation*}
where $\tilde{g}$ is given by the infimal convolution
\begin{equation}
\tilde{g}(z_1,z_2) := \inf_{v\in \mathbb{R}} \Bigl( g(z_1+\sigma(Y_t)
v, z_2) +\mu(Y_t) v \Bigr). \label{infconv}
\end{equation}
To be precise, besides the differences in sign between our convex risk
measures and their monetary concave utility functionals, our
$L^\infty$ framework is in line with the main part of their paper
where they work in $L^\infty$ . However, the result \cite[Theorem
7.17]{KS2} is stated in the framework of $L^2$-BSDEs with Lipschitz
drivers. Their detour to $L^2$ was due to their consideration that
this is the natural framework for BSDEs. We have motivated that we
have to work with quadratic drivers, for which there is yet no
$L^2$-theory, but it is straightforward (though tedious) to check that
their result \eqref{infconv} adapts to our setting due to the
regularity enforced by the admissibility conditions in Definition
\ref{admiss}.

Next, we rewrite the infimum in \eqref{infconv} to get
\[
\tilde{g}(z_1,z_2) = \inf_{u \in \mathbb{R}} \Bigl( g(u,z_2)
-(z_1-u)\lambda(Y_t)\Bigr)\\ 
= -z_1 \lambda(Y_t)
-\hat{g}\bigl(-\lambda(Y_t), z_2\bigr). 
\]
Finally, the uniqueness of the solution of the BSDE
\eqref{transformBSDE} follows again from \cite[Corollary 6]{BH} using the
convexity of the driver, which is implied by the fact that $\hat{g}$
is concave in the second component. Moreover, this entails also that
$\hat{\varrho}$ is a dynamic convex risk measure.
\end{proof}

\begin{remark}
\begin{enumerate}
\item We remark that in view of equation
\eqref{transformBSDE} of the above proposition, the notion of
admissibility could be slightly extended:  it is possible to replace the
lower bound in Definition \ref{admiss}, iii) by 
\begin{equation*}
 c_1 \bigl(f(z_1) -(1+z_2^2)\bigr) \leq  g(z_1,z_2)
\end{equation*}
for an arbitrary real-valued, convex and $1$-coercive
function $f$. This is enough to get existence and uniqueness of
equation \eqref{transformBSDE}, however it would clearly destroy the
nice invariance property of Lemma \ref{admissghat} and we do not adopt
it in the following.
\item As in \eqref{transformBSDE} and the rest of the paper, the first argument of $\hat{g}$ is the negative of the Sharpe ratio.
If we knew {\em a priori} that the Sharpe ratio was small relative to the minimal Lipschitz constant $L$ in Remark \ref{rem26}, then we could allow BSDEs with Lipschitz drivers. However we choose to put restrictions on the driver, such as in Definition \ref{admiss}(iii), rather than on the range of the model parameters.
\end{enumerate}
\end{remark}

\subsection{Indifference Valuation of European Claims}

From the formula \eqref{indiffdef} for the indifference price $P_t$
and Proposition \ref{riskBSDE}, we have that
\begin{equation}
P_t =  \hat{R}_t^{(0)} - \hat{R}_t^{(\xi)}. \label{indiff2}
\end{equation}
From now on we will restrict ourselves to particular bounded payoffs,
namely European put options with strike price $K$ and maturity date
$T$: $\xi = (K-S_T)^+$. Moreover, for the further treatment the substitutions 
\begin{equation}\label{subst}
x:=\log{(S/K)}, \qquad \tau := T-t,
\end{equation}
will be convenient and we introduce the following notation. Denote by
$L_T$ the layer $[0,T]\times \mathbb{R}^2$ and by $Q_{\tau_0,r}$ the
open cylinder above the disk $B(m,r)$ with midpoint $m$, radius $r$
and height $0<\tau_0 \leq T$: $Q_{\tau_0,r} :=]0,\tau_0[ \times
B(m,r)$. Since the location of the midpoint (once fixed) will play no
further role, we skip it in the notation.

The following theorem characterizes the indifference price of a European
put with respect to the dynamic convex risk measure with driver $g$
under the stochastic volatility model \eqref{SVM}. 

\begin{theorem}\label{pricePDE}
The buyer's indifference price of the European put option is given as
\begin{equation}\label{put}
P(\tau,x,y) = \tilde{u}(\tau,y) - u(\tau,x,y)
\end{equation}
where  $u \in C^{1+\beta/2, 2+\beta}(Q_{T,r}) \cap C(L_T)$ for every
cylinder $Q_{T,r}$ is the solution of the semi-linear parabolic PDE 
\begin{equation}\label{riskPDEin}
\left\{
\begin{array}{rl}
-u_\tau +  L u = &\frac{1}{K}\hat{g}\bigl(
-\lambda(y),\rho' K a(y) u_y\bigr); \\ 
u(0,x,y)  = &-(1-e^x)^+, 
\end{array}
\right. 
\end{equation}
with operator $L$ given by
\begin{align}\label{laplacian}
Lu &= \mathcal{M}u -\frac{1}{2}\sigma^2(y) u_x + \bigl(m(y)- \rho a(y) \lambda(y)\bigr)u_y \nonumber \\
\mathcal{M}u &= \frac{1}{2} \sigma^2(y) u_{xx} +  \rho \sigma(y) a(y) u_{x y}+ \frac{1}{2}a^2(y)u_{yy}
\end{align}
and $\tilde{u}$ denotes the ($x$-independent) solution of \eqref{riskPDEin}, with
altered initial condition $\tilde{u}(0,y) = 0$. Moreover there exists a
solution of the Cauchy problem \eqref{riskPDEin} (as well as that with
the altered initial condition) which is the unique classical solution that
is bounded in $L_T$ together with its derivatives. 
\end{theorem}

\begin{proof}
By Ladyzhenskaya {\it et al.} \cite[Theorem V.8.1 and Remark
V.8.1]{LSU}, there exists a classical solution  $v\in C^{1+\beta/2,
    2+\beta}(Q_{T,r}) \cap C(L_T)$  to the semilinear parabolic PDE
\begin{equation}
v_t + \frac{1}{2} \sigma^2(y)S^2 v_{SS} + \frac{1}{2}a(y)^2v_{yy} +
\rho \sigma(y)a(y)Sv_{Sy} + \bigl(m(y)-\rho a(y)
\lambda(y)\bigr) v_y 
= \hat{g}\bigl( -\lambda(y),\rho' a(y) v_y\bigr), \label{vpde}
\end{equation}
with terminal condition $v(T,S,y)=-(K-S)^+$ which is the unique classical solution to this PDE that
is bounded in $L_T$ together with its (first and second order) derivatives. Applying It\^o's formula
to $v(t, S_t, Y_t)$ and defining
\begin{align*}
\bar{Z}_t^1 &:= \sigma(Y_t)S_t v_s(t,S_t,Y_t) + \rho a(Y_t)v_y(t,S_t,Y_t)\\
\bar{Z}_t^2 &:= \rho' a(Y_t) v_y(t,S_t,Y_t),\\
\bar{R}_t &:= v(t,S_t,Y_t)
\end{align*}
shows that $(\bar{R}_t,\bar{Z}_t^1,\bar{Z}_t^2)$ solves the BSDE
\eqref{transformBSDE} for
$(\hat{R}_t^{(\xi)},\hat{Z}_t^1,\hat{Z}_t^2)$ with $\xi=(K-S_T)^+$,
and therefore we identify $\hat{R}_t^{(\xi)}=v(t,S_t,Y_t)$. The
transformation \eqref{subst}, together with $u(\tau, x,y)=v(t,s,y)/K$
leads to the Cauchy problem \eqref{riskPDEin} for $u$.  Finally,
taking zero terminal condition for the PDE \eqref{vpde}, and calling
the solution $\tilde{v}(t,y)$ leads to
$\hat{R}_t^{(0)}=\tilde{v}(t,Y_t)$. Therefore the indifference
price in \eqref{indiff2} is given by $P_t = \tilde{v}(t,Y_t) -
v(t,S_t,Y_t)$, which, in transformed notation, leads to \eqref{put}.
\end{proof}

Here $u$ is the value function of the holder of the put option, and
$\tilde{u}$ is related to the investment (or Merton) problem with
trading only in the underlying stock and money market account.  The nonlinearity in the PDE \eqref{riskPDEin} is in its ``Greek''
$u_y$, that is, the Vega, and enters through the
Legendre transform of the driver $g$ in its first variable. For the
familiar entropic risk measure, $g(z_1,z_2)=\gamma (z_1^2+z_2^2) /2$,
where $\gamma>0$ is a risk-aversion parameter, we have
$\hat{g}(\zeta,z_2) = (\zeta^2/\gamma - \gamma z_2^2) /2$. In this
case, the nonlinearity is as $u_y^2$ (see for example
\cite{benth,sz}).

\subsection{Implied Volatility PDE}\label{volPDE}
Our main goal is to establish an asymptotic expansion of the
indifference price implied volatility in the limit of short
time-to-maturity.  To do so, we now adapt the approach of Berestycki
et al. \cite{BBF} to establish a PDE satisfied by the Black-Scholes
volatility $I(\tau, x,y)$ implied by the indifference pricing.
Therefore we note that in the Black-Scholes model with unit
volatility, the no arbitrage pricing PDE is given by
\begin{equation*}
\left\{
\begin{array}{rl}
-U_\tau +\frac{1}{2}  \bigl(U_{xx} -U_x \bigr) & = 0; \\
U(0,x) & = (1-e^x)^+,
\end{array}
\right. 
\end{equation*}
which can be represented explicitly as 
\begin{equation*}
U(\tau, x) = \Phi\Bigl(-\frac{x}{\sqrt{\tau}}+
\frac{\sqrt{\tau}}{2}\Bigr) - e^x \Phi\Bigl(-\frac{x}{\sqrt{\tau}} -
\frac{\sqrt{\tau}}{2}\Bigr), 
\end{equation*}
where $\Phi$ is the cumulative density function of the standard normal
distribution. Using the scaling properties of the Black-Scholes put
price, the indifference pricing implied volatility $I(\tau,x,y)$ is
hence given by the equation
\begin{equation}\label{BSinversion}
P(\tau,x,y) = U(I^2(\tau,x,y)\tau, x).
\end{equation}
Before we derive a PDE for the implied volatility, we give some a
priori bounds on indifference prices and their associated implied
volatilities. 

 \begin{proposition}\label{pvolbounds}
   Denote by $P^{BS}(\tau,x;\sigma)$
   the Black-Scholes price of the put calculated with constant
   volatility $\sigma$. Then
\begin{equation}\label{pricebounds}
P^{BS}(\tau,x;\underline{\sigma}) \leq P(\tau,x,y) \leq
P^{BS}(\tau,x;\overline{\sigma}) 
\end{equation}
and
\begin{equation}\label{volbounds}
\underline{\sigma} \leq I(\tau,x,y) \leq \overline{\sigma},
\end{equation}
where $\underline{\sigma}$ and $\overline{\sigma}$ are the volatility
  bounds in Assumption \ref{mainassump}.
\end{proposition}
The proof is given in Section \ref{proofs}. 

To derive the PDE for the implied volatility, we plug \eqref{BSinversion} into the
equation \eqref{put} and get after some calculations - detailed in Section \ref{proofs} -  that the implied
volatility $I$ is subject to the nonlinear degenerate parabolic PDE
\begin{equation}
- (\tau I^2)_\tau + \tau I \mathcal{M} I +  I^2 \mathcal{G} \frac{x}{I} - \frac{1}{4} \tau^2 I^2 \mathcal{G}I +\tau q(y) II_y 
= 2\tau  I I_y \frac{\hat{g}\bigl(-\lambda(y),\rho' K
  a(y) \tilde{u}_y \bigr) -
  \hat{g}\bigl(-\lambda(y),\rho' K a(y) u_y \bigr)}{ K
  (\tilde{u}_y - u_y)},\label{Ieqn}
\end{equation}
with the square field type operator
\begin{equation}\label{squarefield}
\mathcal{G}I := \mathcal{M}I^2 - 2I\mathcal{M}I
\end{equation}
and
\begin{equation}\label{q}
q(y) := \rho \sigma(y) a(y) + 2m(y) -2 \rho a(y) \lambda(y).
\end{equation}

To motivate an initial condition for the Cauchy problem, we send
formally $\tau$ to zero in \eqref{Ieqn}. As the ``Vega'' $\vega =
\tilde{u}_y - u_y$ tends also to zero as $\tau\downarrow0$ (this is
shown in Lemma \ref{vanveg}), we observe that the quotient on the right side
of \eqref{Ieqn} 
\[ \frac{\hat{g}\bigl(-\lambda(y),\rho' K
  a(y) \tilde{u}_y \bigr) -
  \hat{g}\bigl(-\lambda(y),\rho' K a(y) u_y \bigr)}{ K
  (\tilde{u}_y - u_y)} \to
\rho'a(y)\hat{g}_{z_2}\bigl(-\lambda(y),0\bigr), \]
which is bounded by the definition of  strictly quadratic drivers and 
Lemma \ref{admissghat}. Dividing by $I^2$, this leads to the formal limit equation
\begin{equation}
\mathcal{G}\frac{x}{I(0,x,y)}=1. \label{initcond}
\end{equation}
\begin{remark}\label{BBFdiff}
Our Cauchy problem is similar to that derived for no arbitrage pricing
implied volatilities in \cite{BBF}, where they have the same equation
\eqref{Ieqn}, but i) without the last term on the left side (which
here is due to the change in measure from physical to a risk-neutral
one); and ii) without the right side term (which here is due to the
dynamic convex risk measure used for indifference pricing). However,
our initial condition \eqref{initcond}, which does not depend on
$\hat{g}$ and the drift of the stochastic volatility model, is exactly
the same as theirs.
\end{remark}

Now we turn this heuristic argument into a precise statement.

\begin{theorem}\label{maintheorem}
The implied volatility function $I(\tau,x,y)$ generated by the
indifference pricing mechanism is the unique solution $I \in
C^{1+\beta/2, 2+\beta}(Q_{T,r}) \cap C(L_T)$ to the following
nonlinear parabolic Cauchy problem
\begin{equation}\label{IPDEind}
- (\tau I^2)_\tau + \tau I \mathcal{M} I +  I^2 \mathcal{G} \frac{x}{I} - \frac{1}{4} \tau^2 I^2\mathcal{G}I +\tau q (y)II_y 
= 2\tau I I_y \frac{\hat{g}\bigl(-\lambda(y),\rho' K
  a(y) \tilde{u}_y \bigr) -
  \hat{g}\bigl(-\lambda(y),\rho' K a(y) u_y \bigr)}{ K
  (\tilde{u}_y - u_y)} 
\end{equation}
where $\mathcal{M}$, $\mathcal{G}$, and $q$ are given by \eqref{laplacian},
\eqref{squarefield}, and \eqref{q} respectively.
The initial condition is given as $I(0,x,y) = x/\psi(x,y)$ where
$\psi$ is the unique viscosity solution of the eikonal equation 
\begin{equation}\label{eikonal}
\left\{
\begin{array}{rl}
\mathcal{G} \psi &=1; \\ 
\psi(0,y) & = 0;\\
\psi(x,y) & > 0 \quad \mbox{ for } x>0.
\end{array}
\right. 
\end{equation}
\end{theorem}
The proof is given in Section \ref{proofs}. 

It is worthwhile to note that indifference prices are not
linear. Indeed, buying double the amount of securities will not lead
to twice the price. In this way also the volatility implied by
indifference prices is quantity-dependent as a consequence of the
appearance of $u_y$ and $\tilde{u}_y$ on the right side of equation
\eqref{IPDEind}.  However, the nonlinearity in quantity is not
observed in the zeroth- and first-order small time-to-maturity
approximation as we show in the following subsection.  Moreover,
deriving the PDE for the indifference price implied volatility for
buying $n$ put options results in the same Cauchy problem as in
Theorem \ref{maintheorem}, where one has only to replace $K$ by $nK$
in every appearance in equation \eqref{IPDEind}.

\subsection{Small-Time Expansion\label{expansions}}
In the short time limit the implied volatility under indifference
pricing is equal to the usual one as calculated by Berestycki
e.a. \cite{BBF}, as the initial conditions are the same (see Remark
\ref{BBFdiff}). The subtleties of the indifference pricing appear
only away from maturity. This can be seen by methods of a formal small-time expansions. We make the {\it Ansatz} of an asymptotic expansion of the
implied volatility:
\begin{equation}\label{expansion}
I(\tau,x,y) = I^0(x,y)\bigl(1 + \tau I^1(x,y) + O(\tau^2)\bigr).
\end{equation}
As seen above, the term $I^0(x,y) = I(0,x,y)$ is given via solution of the eikonal
equation \eqref{eikonal}. To find the PDE for $I^1$, we plug the
expansion \eqref{expansion} into the equation \eqref{IPDEind} and
compare the first order terms for $\tau \to 0$. This leads to the
inhomogeneous linear transport equation
\begin{align}\label{firstPDE}
&2 I^1(x,y)+ \frac{\psi}{2}\Bigl( \sigma^2(y) \psi_x + \rho \sigma(y) a(y) \psi_y\Bigr)I^1_x(x,y) + \frac{\psi}{2}\Bigl( \rho \sigma (y) a(y) \psi_x + a^2(y) \psi_y \Bigr) I^1_y(x,y) -\frac{\psi}{x} \mathcal{M}\frac{x}{\psi} \nonumber\\
 = &- \Bigl(q(y)  - \rho'a(y) \hat{g}_{z_2}\bigl(-\lambda(y),0\bigr)\Bigr)  \frac{\psi_y}{\psi}
\end{align}
where $\psi$ is again the solution of the eikonal equation
\eqref{eikonal}.

It is important to observe that the dependence of
this first order approximation on the risk measure (via its driver
$g$) is given merely by the evaluation at $z_2 =0$ of the derivative
of its Fenchel-Legendre transform $\hat{g}$ with respect to the second
component.

Comparing our PDE to the analogous equation in the arbitrage-pricing setting of \cite{BBF} (who, however, prescribe only the methodology in general and make explicit calculations just in one example), we note the presence of two  additional  terms in  \eqref{firstPDE} - one due to the change to a risk-neutral probability measure and the second due to indifference pricing with a dynamic convex risk measure.

Furthermore in many we can obtain an interior boundary condition for the PDE
at $x=0$ by sending $x$ formally to zero in \eqref{firstPDE}, since by $\psi \to0$ and $\psi_x$ it follows that the first order coefficients vanish. We will show this specific procedure on a concrete example in Section \ref{hwmodel}. Imposing
higher regularity on the coefficients and on $\hat{g}$ one can obtain
also higher order terms in the expansion of the implied
volatility. This is done by using Taylor expansions of $\hat{g}$ (in
the second component) and $\tilde{u}_y$ and $u_y$ (in
$\tau$).

\section{Examples \& Computations}\label{ec}
In this Section, we introduce a family of dynamic risk measures within which to present the effect of risk aversion on implied volatilities, first using the asymptotic approximation in the Hull-White stochastic volatility model, and later using a numerical solution of the quasilinear option pricing PDE.

\subsection{Distorted Entropic Dynamic Convex Risk Measures}
To study the impact
of the driver on the implied volatility, we will now turn to a nicely
parametrized family of risk measures. Therefore we define the
following class of drivers, generating distorted forms of the entropic
risk measure: 
\begin{equation}
g^{\eta, \gamma}(z_1,z_2) := \frac{\gamma}{2}(z_1^2 + z_2^2) +
\eta \gamma z_1 z_2 + \frac{\eta^2 \gamma}{2}z_2^2 =
\frac{\gamma}{2}\Bigl((z_1+\eta z_2)^2 +z_2^2\Bigr). \label{distorted} 
\end{equation}
It is clear, that in the case $\eta = 0$ this is the driver connected
to the classical entropic risk measure, whereas $\eta$ can be seen as
a parameter which describes in which way volatility risk increases
also the risk coming from the tradable asset. As we will see later in Section \ref{numerical}, $\eta$ effectively plays the role of a {\it volatility risk premium}. In the case of the
usual entropic risk measure the driver describes a circle
whose radius is governed by the parameter $\gamma$. In the
distorted case it is now an ellipse where $\eta$
determines additionally the eccentricity.

Turning to the Fenchel-Legendre transform, we have
\begin{equation}
\hat{g}^{\eta, \gamma}(\zeta,z_2) = \frac{1}{2 \gamma}\zeta^2 -
\frac{\gamma}{2} z_2^2 -\eta \zeta\, z_2. \label{ghat}
\end{equation}
Plugging this into \eqref{IPDEind}, we see that the right hand side now reads
\begin{equation*}
\tau I I_y \Bigl( 2 \eta \rho' \lambda(y) a(y) - \gamma K {\rho'}^2a(y)^2(\tilde{u}_y + u_y)\Bigr)
\end{equation*}
and we remark in particular that $\gamma$ scales with $K$ and hence also with the number of securities bought (as mentioned at the end of Section \ref{volPDE}). In particular we see again that the term appearing in the first order approximation of the implied volatility, $\hat{g}^{\eta, \gamma}_{z_2}(-\lambda(y),0)= \eta \lambda(y)$, is independent of $\gamma$.

\subsection{Short-Time Asymptotics for the Hull-White Model}\label{hwmodel}
In the following we look at an example which is an adaption of Example
6.1/6.3 of \cite{BBF}. Let the stochastic volatility
model be given as the Hull-White model 
\begin{equation}\label{SVMEx}
\left\{
\begin{array}{lll}
dS_t &= \mu (Y_t) S_t \,dt +  Y_t S_t \, dW_t^1, \qquad & S_0 = s; \\
dY_t & = \kappa Y_t \, dW_t^2, \qquad & Y_0 = y.
\end{array}
\right. 
\end{equation}
for two independent Brownian motions $W^1$, $W^2$. 
Obviously the model does not fall in the class considered above
because the volatility $\sigma(Y)=Y$ is a geometric Brownian motion
that is not bounded above or away from zero. Nevertheless, by a change
of variables we will see that the results hold. 

Writing down the pricing PDE in the case of the distorted entropic risk measure \eqref{distorted}, we get
\begin{equation*}
-u_\tau + \frac{1}{2}\bigl(y^2 u_{xx} + \kappa^2y^2 u_{yy}\bigr) - \frac{1}{2}  y^2 u_x = \frac{1}{2 \gamma K} \frac{\mu(y)^2}{y^2}- \frac{\gamma K \kappa^2}{2}y^2u_y^2 + \eta \kappa \mu(y) u_y,
\end{equation*}
and one sees that by the time change
$\tau \mapsto \tau y^2$ (the boundary $y= 0$ is not hit when we start with $y_0 >0$ given that $Y_t$ is a geometric Brownian motion) and setting $\tilde{\mu}(y) := \mu(y) / y^2$ the equation becomes 
\begin{equation*}
-u_\tau + \frac{1}{2}\bigl( u_{xx} + \kappa^2 u_{yy}\bigr) - \frac{1}{2}  u_x = \frac{1}{2 \gamma K } \tilde{\mu}(y)^2- \frac{\gamma K \kappa^2}{2} u_y^2 + \eta \kappa\tilde{ \mu}(y) u_y.
\end{equation*}
This equation has a solution (again by  \cite[Theorem V.8.1 and Remark
V.8.1]{LSU}), at least in the case that
$\tilde{\mu}$ is locally $\beta$-H\"{o}lder continuous
(which in turn implies that $\mu(y) = O(y^2)$ as $y \to 0$) .

In absence of global bounds on the volatility we are not able to derive results about existence and uniqueness of solutions for the PDE \eqref{IPDEind} for the implied volatility. Nevertheless we can postulate the small-time expansion to get
\begin{equation}\label{HWeikonal}
\left\{
\begin{array}{rll}
\psi_x^2 + \kappa^2 \psi_y^2 &= 1/y^2\\
\psi(0,y) & = 0\\
\psi(x,y) & >0 \quad \mbox{for } x>0
\end{array}
\right. 
\end{equation}
as the PDE characterizing it's zeroth oder term and
\begin{equation}\label{FO}
\left\{
\begin{array}{lll}
0 &= 2I^1 + y^2 \psi \psi_x I^1_x +\kappa^2y^2\psi \psi_y
I^1_y -\frac{x}{\psi} \mathcal{M} \frac{\psi}{x}  - \eta \kappa \mu(y) \frac{\psi_y}{\psi} \\ 
I^1(0,y) & = \frac{\kappa^2}{12} + \eta \frac{\mu(y)}{2y}
\end{array}
\right. 
\end{equation}
for the first. The interior boundary condition for $I^1$ at $x=0$ follows from the formal asymptotics as $x \to 0$. As derived in \cite{BBF}, the zeroth order term of the expansion is given via the solution of \eqref{HWeikonal}, 
\begin{equation}
\psi(x,y) = \frac{1}{\kappa} \ln{\biggl(\frac{\kappa x }{y} +
  \sqrt{1+\frac{\kappa^2 x^2}{y^2}}\biggr)}. \label{psiformula}
\end{equation}
as $I^0(x,y) = x/\psi(x,y)$, whereas for \eqref{FO} we can guarantee only a solution in the case where $\mu(y)=O(y^3)$ since
$I^0_y/I^0 = -\psi_y/\psi = o(1/y)$ as $y \to 0$. Obviously this means practically that we need an extreme drift in the Hull-White model to compensate the very
volatile volatility process.
However, setting  e.g. $\mu(y) = \mu y^3$ for some constant $\mu$, we are able to solve the PDE \eqref{FO} explicitly by the method of characteristics to get
\begin{equation}
I^1(x,y) = \frac{1}{\psi^2(x,y)}\Biggl(\ln{\biggl(\frac{
    y}{x}\psi(x,y) \Bigl(1+\frac{\kappa^2 x^2}{y^2}\Bigr)^\frac{1}{4}\biggr)} + \eta \frac{ \mu x^2}{2}
\Biggr) . \label{I1formula}
\end{equation}

In Figure \ref{fig:IS}, we rely on the parameter set
\begin{equation*}
\mu=6; \qquad \kappa= 7; \qquad y = 0.3; \qquad \tau = 0.1. 
\end{equation*}

\begin{figure}[htb]
\begin{center}
\includegraphics[width=0.8\linewidth]{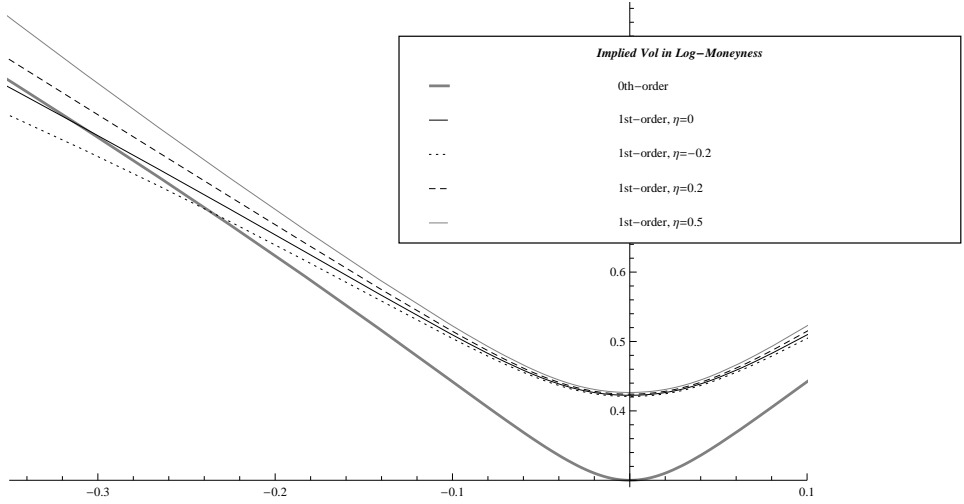}
\caption{Implied volatility in terms of log-moneyness $\log(K/S_0)$ for the Hull-White model: zeroth
  and first order approximation in dependence of $\eta$.}
\label{fig:IS}
\end{center}
\end{figure}

Whereas the parameter $\gamma$ does not appear in the first order
approximation, the distortion parameter $\eta$ has a double effect. On
the one hand side it shifts the smile at the money a small amount, on the other hand
it changes more significantly the wing behavior of the smile, adding to the
asymptotics the term $\eta \frac{\mu
  \kappa^2}{2}\frac{x^2}{(\ln{\vert x \vert})^2}$ (since
$x^2/\psi^2 \sim k^2 x^2/(\ln{\vert x\vert})^2$ as $x \to \pm
\infty$). This changes the whole wing behavior, since $I^0 \sim \kappa
\frac{\vert x \vert}{\ln{\vert x\vert}}$ and $I^1 \to 0$ for $\eta =
0$ as $x \to \pm \infty$. Of course $\eta=0$ corresponds to the first
order term of martingale pricing as \cite{BBF}. Positive $\eta$ (hence
a positive impact of the volatility risk on the risk of the traded
asset) increases the implied volatility and steepens the wings.

\subsection{Calibration Exercise\label{calib}}
We perform a simple calibration exercise using S\&P 500 implied volatilities to show how the approximation \eqref{expansion} could be used to recover some information about the market's risk measure. We work with the short time approximation in the uncorrelated Hull-White model \eqref{SVMEx} of Section \ref{hwmodel} and within the family of distorted dynamic convex risk measure \eqref{distorted}. 

We suppose that the near-the-money option prices are given by expectations under an equivalent martingale measure, specifically the minimal martingale measure. This corresponds to setting the distortion $\eta=0$, and so the short-time asymptotic approximation for implied volatilities $I$ is
\[ I \approx \frac{x}{\psi(x,y)}\left(1+\tau I^1(x,y)\right),
\]
where $\psi(x,y)$ is given by \eqref{psiformula} and $I^1(x,y)$ by \eqref{I1formula} with $\eta=0$. 
This gives
\[ I\approx I_M:=\frac{x}{\psi(x,y)}\left(1+\frac{\tau}{\psi^2(x,y)}\ln{\left(\frac{
    y}{x}\psi(x,y) \left(1+\frac{\kappa^2 x^2}{y^2}\right)^\frac{1}{4}\right)} \right),\]
where we see that the approximation $I_M$ depends only on the parameter $\kappa$ and the current volatility level $y$. Then having fit the liquid implied volatilities to recover estimates of $\kappa$ and $y$, we can write the short time approximation with distortion as 
\begin{equation}
 I\approx I_M +(\mu\eta)\frac{\tau x^3}{2\psi^3(x,y)}, \label{etafit}
 \end{equation}
from which we can estimate the combination $\mu\eta$.

In Figure \ref{calibfit}, we show the result of a fit to S\&P 500 implied volatilities from options with nine days to maturity on October 7th, 2010. We comment that the Hull-White model is not the best choice of stochastic volatility model, and our intention is primarily to demonstrate the procedure. In particular, because our asymptotic formulas are for the uncorrelated model, which always has a smile with minimum implied volatility at the money, we restrict the fit to options with negative log-moneyness, which is mainly out-of-the-money put options. The first part of the least-squares fitting to options with log-moneyness between $-0.06$ and $0$ gives $\kappa\approx 6.6$, and $y\approx0.18$. Then fitting to \eqref{etafit} the options with log-moneyness less than $-0.06$, which we view as less liquid and more reflective of the market's risk measure, gives the estimate $\mu\eta\approx35$. We note that the $\mu$ in this model is expected to be rather large because the composite drift of the stock is $\mu y^3$, and $y<<1$. So the estimate we get from the data of $\mu\eta$ is not unreasonable. In particular it reveals a {\em positive} distortion coefficient $\eta$.
\begin{figure}[htb]
\begin{center}
\includegraphics[width=1.0\linewidth]{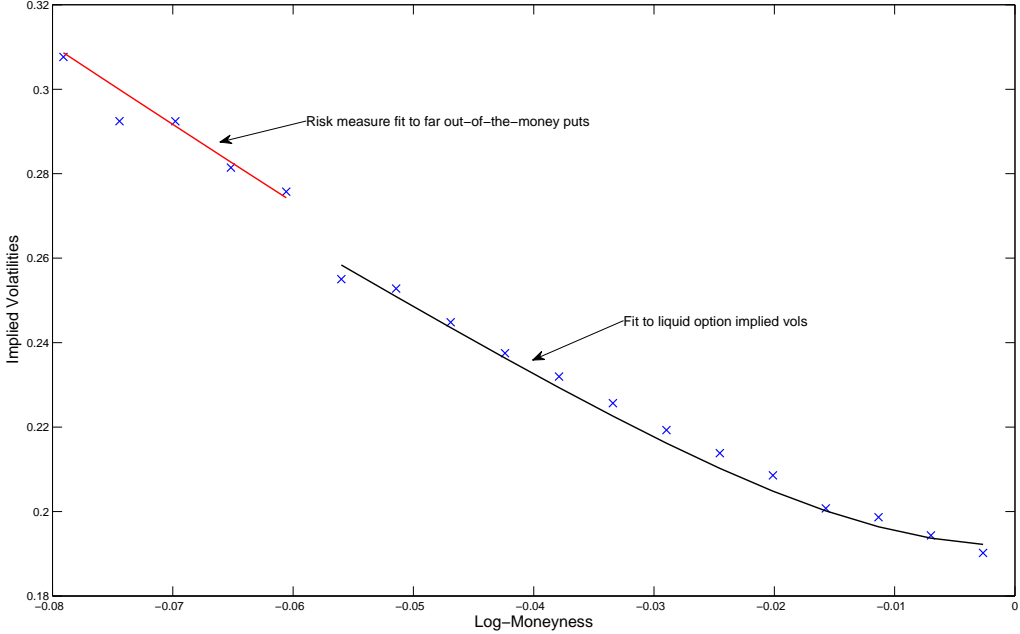}
\caption{Fit of short-time asymptotics to S\&P 500 implied volatility data $9$ days from maturity.}
\label{calibfit}
\end{center}
\end{figure}

\subsection{Numerical Study\label{numerical}}
We consider the buyer's indifference price  of one European put option with respect to the family of {\it distorted entropic risk measures} defined by \eqref{distorted}. We work within the stochastic volatility model \eqref{SVM} and, for the numerical solution, we return to the primitive variables $(t,S,y)$. Denote by $\lop_y$ the generator of the Markov process $Y$:
\[ \lop_y = \half a(y)^2 \frac{\pa^2}{\pa y^2} + m(y) \frac{\pa}{\pa y}, \]
and by $\lop_{S,y}$ the generator of $(S,Y)$:
\[ \lop_{S,y} = \lop_y + \half\sigma(y)^2S^2\frac{\pa^2}{\pa S^2} + \rho a(y)\sigma(y)S\frac{\pa^2}{\pa S\pa y} + \mu(y)\frac{\pa}{\pa S}. \]
From \eqref{indiff2}, the buyer's indifference price of a put option with strike $K$ and expiration date $T$ at time $t<T$ when $S_t=S$ and $Y_t=y$, is given by:
\[ P(t,S,y) = \phi(t,S,y) - \phi_0(t,y), \]
where $\phi$ solves
\begin{eqnarray}
\phi_t + \left(\lop_{S,y} - \rho a(y)\lambda(y)\frac{\pa}{\pa y}\right)\phi &=& -\hat{g}(-\lambda(y),-\rho'a(y)\phi_y), \label{phiPDE} \\
\phi(T,S,y) & = & (K-S)^+,\nonumber
\end{eqnarray}
and $\phi_0(t,y)$ solves the same PDE without the $S$-derivatives and with zero terminal condition. Note that $\phi=-v$, where $v$ was the solution to the PDE problem in \eqref{vpde}, and $\phi_0=-\tilde{v}$ which was introduced in the proof of Theorem \ref{pricePDE}. 

As $\hat{g}$ is given by \eqref{ghat}, we can re-write \eqref{phiPDE} as 
\begin{equation} 
\phi_t + \left(\lop_{S,y} - (\rho +\eta\rho')a(y)\lambda(y)\frac{\pa}{\pa y}\right)\phi = -\frac{\lambda^2(y)}{2\gamma} + \half(1-\rho^2)\gamma a(y)^2\phi_y^2. \label{phiPDE2}
\end{equation}
This shows that $\eta$ plays the role of a {\it volatility risk premium} in that it enters as a drift adjustment for the volatility-driving process $Y$. However the nonlinearity of the PDE is through a quadratic term in $\phi_y$, as in the case of the entropic risk measure. 

Moreover, introducing the transformation
\[ \phi_0(t,y) = -\frac{1}{\gamma(1-\rho^2)}\log f(t,y), \]
leads to the {\it linear} PDE problem for $f$:
\begin{eqnarray}
f_t + \left(\lop_{y} - (\rho +\eta\rho')a(y)\lambda(y)\frac{\pa}{\pa y}\right)f - \half\lambda^2(y)(1-\rho^2)f & = & 0,\label{fPDE}\\
f(T,y) & = & 1.
\end{eqnarray}
Therefore the indifference price is given by
\[ P(t,S,y) = \phi(t,S,y) + \frac{1}{\gamma(1-\rho^2)}\log f(t,y). \]

In the numerical solutions,  we take the volatility-driving
process $(Y_t)$ to be an Ornstein-Uhlenbeck process with the dynamics:
\[ dY_t = \alpha(m-Y_t)\,dt + \nu\sqrt{2\alpha}\,\bigl(\rho \, dW_t^1 +\rho'\,
dW_t^2\bigr), \]
and we choose a function $\sigma(Y)$ that gives realistic volatility
characteristics. For the OU process $(Y_t)$, the rate of mean-reversion is $\alpha$, the long-run
mean-level is $m$ and the long-run variance is
$\nu^2$. For the computations, we will take $\alpha=5$, $m=0$, $\nu^2=1$, $\rho=-0.2$ and
\[ \sigma(y) = \frac{0.7}{\pi}(\arctan (y-1)+\pi/2) + 0.03, \]
so that $\sigma(m)=0.2050$. The parameter $\nu$ measures
approximately the standard deviation of \vol fluctuations. The
values are chosen such that the interval of width one standard deviation
for $Y$ under its stationary distribution is $(-1,1)$, and this translates roughly to the interval $(0.13,0.38)$ for volatility $\sigma$. The two
standard deviation interval for \vol is approximately
$(0.10,0.56)$.

We first solve the quasilinear PDE \eqref{phiPDE2} for $\phi$ using implicit finite-differences on the linear part, and explicit on the nonlinear part. Then we solve the linear PDE problem \eqref{fPDE} for $f$. This procedure is done for a current stock price $S_0=100$ and $\sigma(Y_0)=0.223$, calculating the solutions for European put options expiring in three months for strikes  $K$ in the range $[70,110]$ for different values of the distortion parameter $\eta$ and the risk-aversion parameter $\gamma>0$.

Figure \ref{fig:numerical1} reveals a more complex picture regarding the effect of $\eta$ away from the short maturity asymptotic approximation. We see, as in Figure \ref{fig:IS} from the asymptotics, increasing $\eta$ increases the skew slope; however it also shifts down the levels of implied volatility around the money (as opposed to the opposite effect we saw in Figure \ref{fig:IS}). 

Figure \ref{fig:numerical2} shows, as we would expect, that increasing risk aversion $\gamma$ decreases the implied volatility skew which comes from the indifference price of the buyer who is willing to pay less for the risk of the option position. It also has a secondary effect of flattening the skew out of the money.

\begin{figure}[htbp]
\begin{center}
\includegraphics[width=0.6\linewidth]{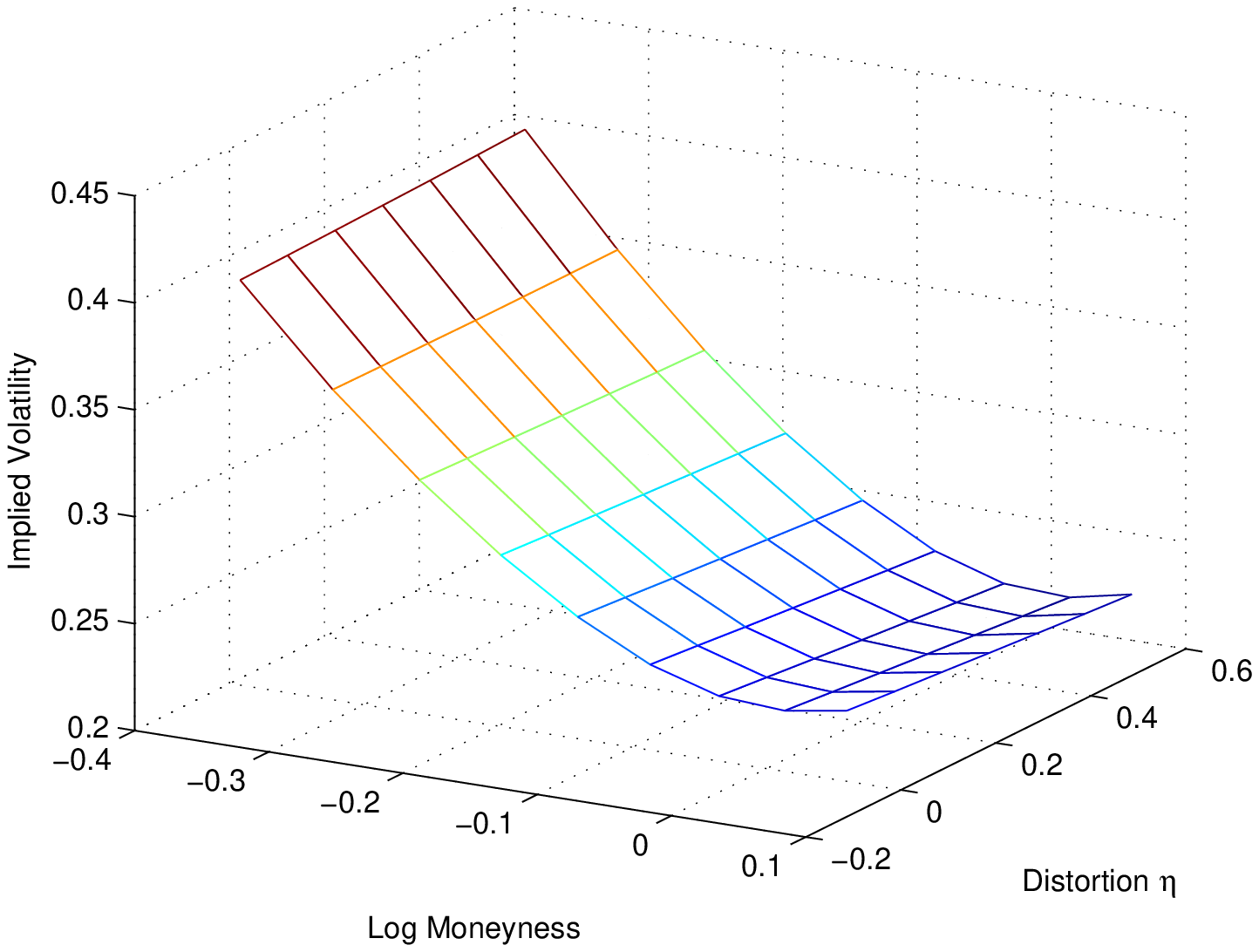}\includegraphics[width=0.5\linewidth]{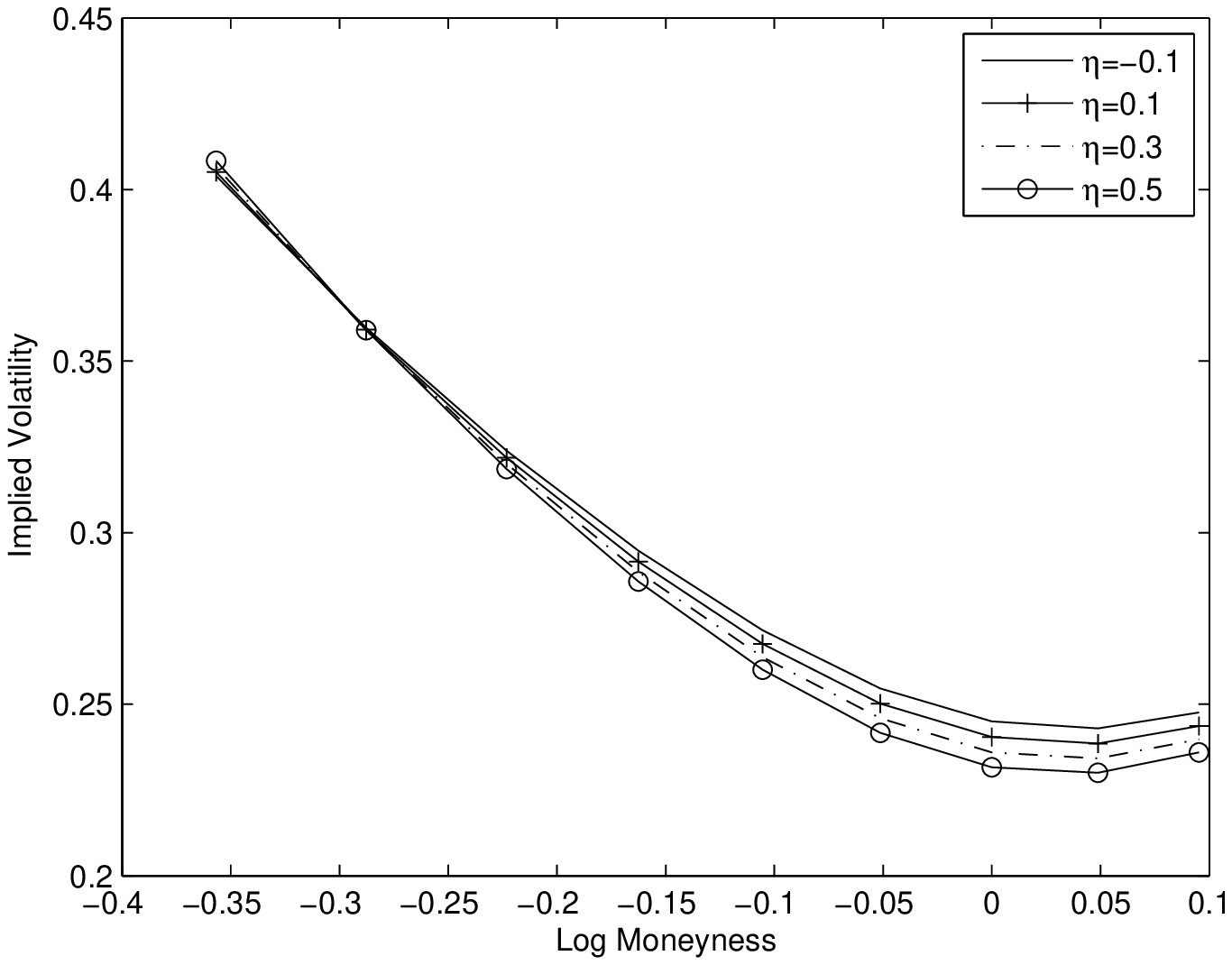}
\caption{Implied volatility from the arctangent stochastic volatility model in terms of log-moneyness $\log(K/S_0)$ and $\eta$ with fixed $\gamma=0.5$.} 
\label{fig:numerical1}
\end{center}
\end{figure}

\begin{figure}[htbp]
\begin{center}
\includegraphics[width=0.6\linewidth]{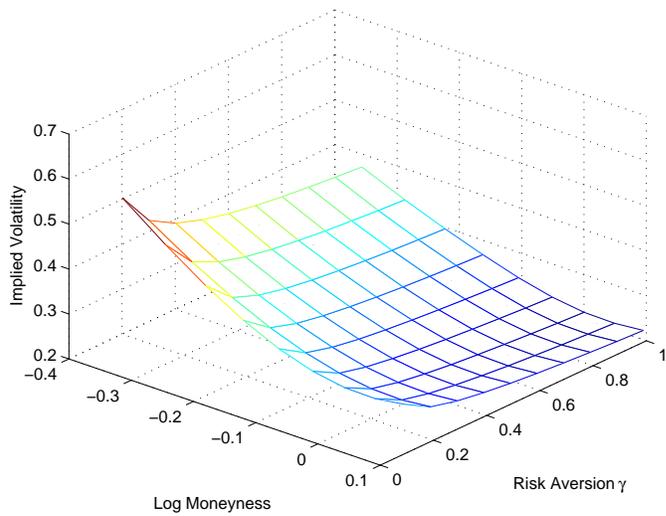}\includegraphics[width=0.5\linewidth]{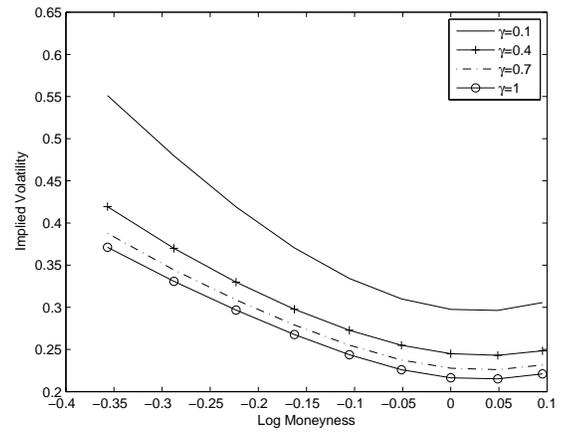}
\caption{Implied volatility from the arctangent stochastic volatility model in terms of log-moneyness $\log(K/S_0)$ and $\gamma$ with fixed $\eta=0.2$.} 
\label{fig:numerical2}
\end{center}
\end{figure}

\section{Conclusion}\label{conc}
We have derived a nonlinear PDE for the implied volatility from
indifference pricing with respect to dynamic convex risk measures
defined by BSDEs under diffusion stochastic volatility models. 
Our asymptotic analysis has highlighted the principal effect of the
risk measure on option implied volatility at short maturities, namely
through the appearance of $\hat{g}_{z_2}$ in the first order
correction solving \eqref{firstPDE}. 

In the example of Section \ref{hwmodel}, this translates explicitly to
a steepening effect on the implied volatility smile from the
distortion parameter $\eta$. Numerical computations confirm this away from short maturity too, as well as quantifying the effect of risk aversion on the level of implied volatilities.

The analysis can be used to infer some information
about the driver, for example $\eta$ and $\gamma$ in the family \eqref{distorted},
from market implied volatilities. This is illustrated using the short time asymptotics for the Hull-White model on S\&P 500 implied volatilities in Section \ref{calib}.

\clearpage

\section{Proofs}\label{proofs}

\paragraph{Price- and Volatility Bounds}

\begin{lemma}[{\bf First Price Comparison}]\label{pricecomp}
Suppose that $\underline{P}= \tilde{\underline{u}} - \underline{u}$
and $\overline{P}= \tilde{\overline{u}} - \overline{u}$ for some
$\underline{u}$, $\tilde{\underline{u}}$, $\overline{u}$,
$\tilde{\overline{u}} \in  C^{1, 2,}(Q_{\tau_0,r}) \cap
C(\overline{Q_{\tau_0,r}}) $, $0<\tau_0 \leq T$, satisfy 
\begin{align}
\overline{P}_\tau &\geq  L \overline{P}
-\frac{1}{K}\Bigl(\hat{g}\bigl( -\lambda(y),\rho' K
a(y)\tilde{ \overline{u}_y}\bigr) -\hat{g}\bigl(
-\lambda(y),\rho' K a(y) \overline{u}_y\bigr)\Bigr),
\quad \mbox{ in }  Q_{\tau_0,r}\label{superint}\\ 
\underline{P}_\tau &\leq  L \underline{P}
-\frac{1}{K}\Bigl(\hat{g}\bigl( -\lambda(y),\rho' K
a(y) \tilde{\underline{u}_y}\bigr) -\hat{g}\bigl(
-\lambda(y),\rho' K a(y) \underline{u}_y\bigr)\Bigr),
\quad \mbox{ in }  Q_{\tau_0,r}  
\end{align}
as well as
\begin{equation}\label{superbound}
\overline{P}\geq\underline{P} \quad \mbox{ on }\quad \{0\} \times
\mathbb{R}^2 \cap \overline{Q_{\tau_0,r}}\quad \mbox { and }\quad
]0,\tau_0] \times \mathbb{R}^2 \cap \partial Q_{\tau_0,r} . 
\end{equation}
Then $\overline{P}\geq\underline{P}$ on $Q_{\tau_0,r}$.
\end{lemma}
\begin{proof}
Even if the form of this comparison principle for sub-
and superprices seems to be quite unusual, the proof follows directly
along the lines of Friedman \cite[Theorem 2.16]{FriePDEPT} since the
functions $\hat{g}$ contains no second derivatives. To be precise:
this argument leads a version where the inequalities in
\eqref{superint}, \eqref{superbound} and the conclusion are
strict. But setting $\overline{P}^\varepsilon = \overline{P} +
\varepsilon (1+\tau)$ one gets a strict superprice and sending
$\varepsilon$ to zero yields the stated version.
\end{proof}

\noindent{\bf Proof of Proposition \ref{pvolbounds}:} 
\begin{proof}
To prove \eqref{pricebounds}, we intend to invoke the above comparison
principle for the price process given in Theorem \ref{pricePDE} since
it is clear that the Black-Scholes pricing functions are sub-
resp. supersolutions of the PDE. Unfortunately we have the
indifference price only as solution of a Dirichlet problem which does
not give rise to directly comparable lateral boundary conditions. thus
we have to alter the argument a bit.

 Denote for $N \in \mathbb{N}$ by $u^{N,\underline{\sigma}}$ the
 solution of the initial/boundary-value problems 
\begin{equation*}
\left\{
\begin{array}{rl}
-u^{N,\underline{\sigma}}_\tau +  L u^{N,\underline{\sigma}} = &\frac{1}{K}\hat{g}\bigl( -\lambda(y),\rho' K a(y) u^{N,\underline{\sigma}}_y\bigr); \\
u^{N,\underline{\sigma}}(0,x,y)  = &-(1-e^x)^+;\\
u^{N,\underline{\sigma}}(\tau,x,y)\big\vert_{\partial B(0,N)} =
&-P^{BS}(\tau,x;\underline{\sigma}). 
\end{array}
\right. 
\end{equation*}
By a classical argument \cite[Section V.\S8]{LSU}, we can extract a
subsequence $u^{N_k,\underline{\sigma}}$ of $u^{N,\underline{\sigma}}$
such that $u^{N_k,\underline{\sigma}}$ converges together with its
derivatives to $u$ and it's derivatives pointwise in $L_T$. The same is true for 
$\tilde{u}^{N}$ given by
\begin{equation*}
\left\{
\begin{array}{rl}
-\tilde{u}^N_\tau +  L \tilde{u}^N= &\frac{1}{K}\hat{g}\bigl(
-\lambda(y),\rho' K a(y) \tilde{u}^N_y\bigr); \\ 
\tilde{u}^N(0,x,y)  = &0;\\
\tilde{u}^N(\tau,x,y)\big\vert_{\partial B(0,N)} = &0.
\end{array}
\right. 
\end{equation*}
Thus $P^{N,\underline{\sigma}}(\tau,x,y) =
u^{N,\underline{\sigma}}(\tau,x,y)-\tilde{u}(\tau,x,y)$ satisfies 
\begin{equation}
\left\{
\begin{array}{rl}\label{PNPDE}
-P^{N,\underline{\sigma}}_\tau +  L P^{N,\underline{\sigma}} = &\frac{1}{K}\hat{g}\bigl( -\lambda(y),\rho' K a(y) \tilde{u}^N_y\bigr)-\frac{1}{K}\hat{g}\bigl( -\lambda(y),\rho' K a(y) u^{N,\underline{\sigma}}_y\bigr); \\
P^{N,\underline{\sigma}}(0,x,y)  = &(1-e^x)^+;\\
P^{N,\underline{\sigma}}(\tau,x,y)\big\vert_{\partial B(0,N)} = &P^{BS}(\tau,x;\underline{\sigma}).
\end{array}
\right. 
\end{equation}
and $P^{N,\underline{\sigma}} \to P$ along a subsequence.\\
Noting that $P^{BS}(\tau,x;\underline{\sigma})$ is a subprice on every
cylinder  $Q_{T,N}$ by writing it in the odd form
$P^{BS}(\tau,x;\underline{\sigma}) =
0-(-P^{BS}(\tau,x;\underline{\sigma}))$ to satisfy the comparison
principle of Lemma \ref{pricecomp}, we have
$P^{BS}(\tau,x;\underline{\sigma}) \leq P^{N,\underline{\sigma}}$ on
$Q_{T,N}$ and hence in the limit $P^{BS}(\tau,x;\underline{\sigma})
\leq P$. The other direction of Theorem \ref{pricePDE} is proved, of
course, by the same argument using $P^{BS}(\tau,x;\overline{\sigma})$
as superprice. 

Finally a reformulation of the achieved result reads
\begin{equation*}
U(\underline{\sigma}^2\tau,x) \leq U(I(\tau,x,y)^2\tau,x) \leq
U(\overline{\sigma}^2\tau,x) 
\end{equation*}
and so the strict monotonicity of $U$ with respect to the first variable yields \eqref{volbounds}.
\end{proof}

\paragraph{Deriving the PDE for the Implied Volatility}
{${}$}\\

To derive the PDE for the implied volatility, we note that from \eqref{put} and \eqref{riskPDEin}, it follows that the indifference price $P$ obeys the equation
\begin{equation}\label{mixedPDE}
-P_\tau + LP = \frac{1}{K}\Bigl(\hat{g}\bigl(
-\lambda(y),\rho' K a(y) \tilde{u}_y\bigr) - \hat{g}\bigl(
-\lambda(y),\rho' K a(y) u_y\bigr) \Bigr).
\end{equation}
Substituting $P(\tau,x,y) = U(I^2(\tau,x,y)\tau,x)$ we note that the derivatives relate
\begin{align*}\label{derivatives}
P_\tau &= U_\tau \cdot (\tau I^2)_\tau;\qquad P_x = U_\tau \cdot 2 I I_x \tau + U_x; \qquad P_y = U_\tau \cdot 2 I I_y \tau;\\
P_{xx} &= U_{\tau \tau} \cdot 4 I^2 I_x^2 \tau^2 + U_\tau \cdot 2 I_x^2 \tau + U_\tau \cdot 2 I I_{xx} \tau + U_{\tau x} \cdot 4 I I_x \tau + U_{xx}; \\
P_{xy} &= U_{\tau \tau} \cdot 4 I^2 I_x I_y \tau^2 + U_\tau \cdot 2 I_x I_y \tau + U_\tau \cdot 2 I I_{xy} \tau + U_{\tau x} \cdot 2 I I_y \tau \\
P_{yy} &= U_{\tau \tau} \cdot 4 I^2 I_y^2 \tau^2 + U_\tau \cdot 2 I_y^2 \tau + U_\tau \cdot 2 I I_{yy} \tau.
\end{align*}
Plugging this into \eqref{mixedPDE} and dividing by $U_\tau$, we derive the equation \eqref{IPDEind} by noting that
\begin{equation}
\frac{U_{xx} - U_x}{U_\tau} = 2; \qquad \frac{U_{\tau \tau}}{U_\tau} = \frac{x^2}{2 I^4 \tau^2} - \frac{1}{2 I \tau} - \frac{1}{8}; \qquad \frac{U_{\tau x}}{U_\tau} = - \frac{x}{I \tau} + \frac{1}{2}; \qquad \frac{1}{U_\tau} = \frac{2 I I_y \tau}{\tilde{u}_y - u_y}.
\end{equation}
The initial condition follows by observing that Vega vanishes as $\tau$ tends to zero:
\begin{lemma}\label{vanveg}
It holds that $\vega = \tilde{u}_y - u_y \to 0$ uniformly on compacts
as $\tau \to 0$. 
\end{lemma}
\begin{proof}
Choose the cylinder $Q_{T,r}$ in the layer $L_T$ such that the compact set
is contained. Thus $u$, $\tilde{u} \in C^{1+\beta/2, 2+\beta}(Q_{T,r})
\cap C(L_T)$ implies that $u_y$ and $\tilde{u}_y$ are
$\beta/2$-H\"{o}lder continuous with some H\"{o}lder constant c,
whence
\begin{align*}
\vert \vega(\tau,x,y) \vert &= \vert \tilde{u}_y(\tau,y)-u_y(\tau,x,y) \vert \\
&\leq \vert \tilde{u}_y(\tau,y)-\tilde{u}_y(0,y) \vert + \vert
u_y(\tau,x,y)-u_y(0,x,y) \vert \leq 2 c \tau^\frac{\beta}{2} \to 0 
\end{align*}
as $\tau \to 0$ since $\tilde{u}_y(0,y)$ and $u_y(0,x,y)$ exist and
are equal to zero by the definition of the initial
conditions.
\end{proof}

\paragraph{Implied Volatility - Proof of the Main Theorem}

\begin{lemma}[{\bf Second Price Comparison}]\label{secondprice}
Recall that $u$ is the solution of the Cauchy problem \eqref{riskPDEin} and $\tilde{u}$ of the same problem with initial condition equal to zero. Suppose that  $\overline{P}$, $\underline{P} \in  C^{1, 2}(Q_{\tau_0,r}) \cap C(\overline{Q_{\tau_0,r}}) $, $0<\tau_0 \leq T$, satisfy
\begin{align}
\overline{P}_\tau &\geq  L \overline{P} -\frac{1}{K}\Bigl(\hat{g}\bigl( -\lambda(y),\rho' K a(y)\tilde{u_y}\bigr) -\hat{g}\bigl( -\lambda(y),\rho' K a(y) u_y\bigr)\Bigr), \quad \mbox{ in }  Q_{\tau_0,r}\label{superint2}\\
\underline{P}_\tau &\leq  L \underline{P} -\frac{1}{K}\Bigl(\hat{g}\bigl( -\lambda(y),\rho' K a(y) \tilde{u_y}\bigr) -\hat{g}\bigl( -\lambda(y),\rho' K a(y) u_y\bigr)\Bigr), \quad \mbox{ in }  Q_{\tau_0,r} 
\end{align}
as well as
\begin{equation}\label{superbound2}
\overline{P}\geq P \geq \underline{P} \quad \mbox{ on } \quad
]0,\tau_0] \times \mathbb{R}^2 \cap \partial Q_{\tau_0,r} , 
\end{equation}
and
\begin{equation*}
\overline{P}(0,x,y) = \underline{P}(0,x,y) = P(0,x,y) = (1-e^x)^+
\end{equation*}
then $\overline{P}\geq P \geq \underline{P}$ on $Q_{\tau_0,r}$.
\end{lemma}
\begin{proof}
We note that inequality \eqref{superint2} implies
$(\overline{P}-P)_\tau \geq L(\overline{P}-P)$ which implies
together with the lateral bound $(\overline{P}-P)\geq 0$ on
$]0,\tau_0] \times \mathbb{R}^2 \cap \partial Q_{\tau_0,r}$ and the
initial condition $\overline{P}(0,x,y)-P(0,x,y)=0$ that
$\overline{P}\geq P$ on $Q_{\tau_0,r}$ by the classical comparison
principle. The second inequality is proved in the same way.
\end{proof}

\begin{lemma}[{\bf Volatility Comparison}]\label{volcomp}
Suppose that $\overline{I}$, $\underline{I} \in  C^{1,
  2}(Q_{\tau_0,r}) \cap C(\overline{Q_{\tau_0,r}}) $, $0<\tau_0 \leq T$,
satisfy 
\begin{equation*}
(\tau \underline{I}^2)_\tau \leq  \tau \underline{I}\mathcal{M}\underline{I} + \underline{I}^2 \mathcal{G}\underline{I} + \tau q(y)\underline{I}\underline{I}_y
-2 \tau \underline{I} \underline{I}_y\Biggl( \rho a(y) \lambda(y) + \frac{\hat{g}\bigl(-\lambda(y),\rho' K a(y) \tilde{u}_y \bigr) - \hat{g}\bigl(-\lambda(y),\rho' K a(y) u_y \bigr)}{ K (\tilde{u}_y - u_y)}\Biggr)
\end{equation*}
resp.
\begin{align*}
(\tau \overline{I}^2)_\tau \geq \tau \overline{I}\mathcal{M}\overline{I} + \overline{I}^2 \mathcal{G}\overline{I} + \tau q(y)\overline{I}\overline{I}_y-2 \tau \overline{I} \overline{I}_y\Biggl( \rho a(y) \lambda(y) + \frac{\hat{g}\bigl(-\lambda(y),\rho' K a(y) \tilde{u}_y \bigr) - \hat{g}\bigl(-\lambda(y),\rho' K a(y) u_y \bigr)}{ K (\tilde{u}_y - u_y)}\Biggr)
\end{align*}
in $Q_{\tau_0,r}$ together with the lateral comparison
\begin{equation*}
\underline{I}(\tau,x,y) \leq I(\tau,x,y) \leq \overline{I}(\tau,x,y) \quad \mbox{ on }\quad   ]0,\tau_0] \times \mathbb{R}^2 \cap \partial Q_{\tau_0,r}
\end{equation*}
and the initial growth condition
\begin{equation}\label{growth}
\lim_{\tau \to 0}\tau \underline{I}^2(\tau,x,y) = \lim_{\tau \to 0}\tau \overline{I}^2(\tau,x,y) = 0 \quad \mbox{ on } \quad \{0\} \times \mathbb{R}^2 \cap \overline{Q_{\tau_0,r}}.
\end{equation}
Then it holds that
\begin{equation*}
\underline{I}(\tau,x,y) \leq I(\tau,x,y) \leq \overline{I}(\tau,x,y)
\quad \mbox{ in }\quad  Q_{\tau_0,r}. 
\end{equation*}
\end{lemma}
\begin{proof}
Define first $\overline{P}(\tau,x,y) :=
U(\overline{I}^2(\tau,x,y)\tau,x)$ and $\underline{P}(\tau,x,y) :=
U(\underline{I}^2(\tau,x,y)\tau,x)$. Then by the same calculation as
in the derivation of the PDE \eqref{IPDEind} of the implied volatility in the paragraph above
we get
\begin{align*}
\overline{P}_\tau &\geq  L \overline{P} -\frac{1}{K}\Bigl(\hat{g}\bigl( -\lambda(y),\rho' K a(y)\tilde{u_y}\bigr) -\hat{g}\bigl( -\lambda(y),\rho' K a(y) u_y\bigr)\Bigr), \quad \mbox{ in }  Q_{\tau_0,r}\\
\underline{P}_\tau &\leq  L \underline{P} -\frac{1}{K}\Bigl(\hat{g}\bigl( -\lambda(y),\rho' K a(y) \tilde{u_y}\bigr) -\hat{g}\bigl( -\lambda(y),\rho' K a(y) u_y\bigr)\Bigr), \quad \mbox{ in }  Q_{\tau_0,r} 
\end{align*}
as well as the lateral boundary condition
\begin{equation*}
\underline{P}(\tau,x,y) \leq P(\tau,x,y) \leq \overline{P}(\tau,x,y)
\quad \mbox{ in }\quad Q_{\tau_0,r}. 
\end{equation*}
Moreover, the growth condition \eqref{growth} implies by the
continuity of $U$ that $\overline{P}(0,x,y) = \underline{P}(0,x,y) =
P(0,x,y) = (1-e^x)^+$. Thus we can use Lemma \ref{secondprice} to
infer $\underline{P}(\tau,x,y) \leq P(\tau,x,y) \leq
\overline{P}(\tau,x,y)$ in $Q_{\tau_0,r}$ and the strict monotonicity of the
function $U$ in the first component yields the result.
\end{proof}

\noindent{\bf Proof of Theorem \ref{maintheorem}:} 
\begin{proof}
Remark first that if there
exists a solution to the PDE with some fixed initial condition, it has to be unique by the smoothness and
strict monotonicity of $U$ (and the boundedness of $I$ proven in Proposition \ref{pvolbounds}) since otherwise the
solution of the pricing PDE (Theorem \ref{pricePDE}) would not be
unique. By the same reasoning we get also $I \in  C^{1+\beta/2,
  2+\beta}(Q_{T,r}) \cap C(L_T)$. Moreover, the eikonal equation
\eqref{eikonal} has a unique viscosity solution as proved in
\cite[Section 3.2]{BBF}. To prove the theorem we will hence show that the solution of the eikonal equation is the only possible initial condition, i.e. that any solution of the PDE \eqref{IPDEind} has the eikonal equation as it's small time limit. More precisely we will show that
there exist parametrized families of (time-independent) local super-
and subsolutions of \eqref{IPDEind} which converge locally uniformly to the eikonal
equation. This is done quite similar as in \cite[Section 3.4.]{BBF},
using an adapted vanishing viscosity method. However, in our setting
the bounds on the volatility  $\sigma$ enable us to simplify the proof
and circumvent some obscurities in the local volatility argument in
\cite{BBF}. 

Define $\overline{I}^{\varepsilon, \delta}(x,y)$ for $\varepsilon$,
$\delta>0$ as the solution of 
\begin{equation}\label{logsubsteq}
\left\{
\begin{array}{rl}
-\delta =& - \bigl(\overline{I}^{\varepsilon, \delta}\bigr)^2
+\bigl(\overline{I}^{\varepsilon, \delta}\bigr)^2 \mathcal{G} \frac{x}{\overline{I}^{\varepsilon, \delta}}+ \varepsilon \Delta
\bigl(\ln{(\overline{I}^{\varepsilon, \delta})}\bigr)    ; \\ 
\overline{I}^{\varepsilon, \delta} \big\vert_{\partial B(m,r)} = &\overline{\sigma}.
\end{array}
\right. 
\end{equation}
where $B(m,r)$ is an arbitrary disk. We will show that for $r$,
$\delta$, $\varepsilon >0$ there exists a solution to this equation
and for fixed $r$ and $\delta$ there exist $\varepsilon_0>0$,
$\tau_0>0$ such that for all $0<\varepsilon< \varepsilon_0$ this is a
supersolution of \eqref{IPDEind} in $Q_{\tau_0,r}$. Moreover, we show
that $\overline{I}^{\varepsilon, \delta} \to I^0$ locally uniformly as
we send fist $\varepsilon$ and then $\delta$ to zero.

{\it First step:} Making the change of variables $w: =
\ln{\bigl(\overline{I}^{\varepsilon, \delta}\bigr)}$ one gets
\begin{align*}
\bigl(\overline{I}^{\varepsilon, \delta}\bigr)^2 \mathcal{G} \frac{x}{\overline{I}^{\varepsilon, \delta}}&= \biggl(\overline{I}^{\varepsilon, \delta}\biggr)^2 \Biggl(\sigma^2(y)\biggl(\frac{x}{\overline{I}^{\varepsilon,\delta}}\biggr)_x^2 + 2 \rho \sigma(y) a(y)\biggl(\frac{x}{\overline{I}^{\varepsilon,
  \delta}}\biggr)_x\biggl(\frac{x}{\overline{I}^{\varepsilon, \delta}}\biggr)_y +  a^2(y) \biggl(\frac{x}{\overline{I}^{\varepsilon,
  \delta}}\biggr)_y^2\Biggr) \nonumber\\ 
&=\sigma^2(y) (1-xw_x)^2 -2\rho \sigma(y) a(y) x(1-x w_x)w_y + a^2(y)x^2w_y^2 =: \tilde{\mathcal{G}}w
\end{align*}
and the equation \eqref{logsubsteq} becomes
\begin{equation*}
\left\{
\begin{array}{rl}
-\delta = &- e^{2w} + \tilde{\mathcal{G}}w + \varepsilon \Delta w ; \\
w \vert_{\partial B(m,r)} = &\ln{\overline{\sigma}}.
\end{array}
\right. 
\end{equation*}
which admits a solution $w \in C^{2+\beta}(B(m,r))$ \cite[Theorem
4.8.3]{LU} which is unique for sufficiently small $r$ \cite[Theorem
4.2.1]{LU}. 

{\it Second Step:} By the H\"{o}lder property of the derivatives of
$w$ resp. $\overline{I}^{\varepsilon, \delta}$ implying the boundedness of the functions on the cylinder (as well as the boundedness of $u_y$ and $\tilde{u}_y$ and the
differentiability of $\hat{g}$) we can conclude that there exist 
constants $c_1$-$c_4$ solely depending on $r$ such that  
\begin{align*}
 - \frac{1}{4} \tau^2 \bigl(\overline{I}^{\varepsilon,\delta}\bigr)^2\mathcal{G} \overline{I}^{\varepsilon,\delta} &\leq c_1(r) \tau^2\\
\tau \overline{I}^{\varepsilon,\delta} \mathcal{M} \overline{I}^{\varepsilon,\delta} &\leq c_2(r) \tau\\
\tau\biggl( q (y)- 2 \frac{\hat{g}\bigl(-\lambda(y),\rho' K
  a(y) \tilde{u}_y \bigr) -
  \hat{g}\bigl(-\lambda(y),\rho' K a(y) u_y \bigr)}{ K
  (\tilde{u}_y - u_y)} \biggr) \overline{I}^{\varepsilon,\delta} \overline{I}^{\varepsilon,\delta}_y &\leq c_3(r)\tau\\
-\varepsilon \Delta \bigl(\ln{(\overline{I}^{\varepsilon, \delta})}\bigr) &\leq \varepsilon c_4(r)
\end{align*}
in $B(m,r)$. We can conclude that
\begin{align*}
\bigl((\tau( \overline{I}^{\varepsilon, \delta})^2\bigr)_\tau = & (\overline{I}^{\varepsilon, \delta})^2 = \delta + I\bigl(\overline{I}^{\varepsilon,\delta}\bigr)^2 \mathcal{G} \frac{x}{\overline{I}^{\varepsilon,\delta}}  + \varepsilon \Delta \bigl(\ln{(\overline{I}^{\varepsilon, \delta})}\bigr)\\
\geq & \tau \overline{I}^{\varepsilon,\delta} \mathcal{M} \overline{I}^{\varepsilon,\delta} +  \bigl(\overline{I}^{\varepsilon,\delta}\bigr)^2 \mathcal{G} \frac{x}{\overline{I}^{\varepsilon,\delta}} - \frac{1}{4} \tau^2 \bigl(\overline{I}^{\varepsilon,\delta}\bigr)^2\mathcal{G} \overline{I}^{\varepsilon,\delta} +\tau q (y) \overline{I}^{\varepsilon,\delta} \overline{I}^{\varepsilon,\delta}_y \\
&- 2\tau I^{\varepsilon,\delta} I^{\varepsilon,\delta}_y \frac{\hat{g}\bigl(-\lambda(y),\rho' K
  a(y) \tilde{u}_y \bigr) -
  \hat{g}\bigl(-\lambda(y),\rho' K a(y) u_y \bigr)}{ K
  (\tilde{u}_y - u_y)} \\
& + \delta - \Bigl( c_1(r) \tau^2 + c_2(r) \tau +
c_3(r)\tau + \varepsilon c_4(r) \Bigr), 
\end{align*}
thus for given $\delta>0$ and $r>0$ we can find indeed positive bounds on $\varepsilon_0$, $\tau_0$ such that $ \overline{I}^{\varepsilon, \delta}$ is a supersolution of \eqref{IPDEind} for $0<\varepsilon \leq \varepsilon_0$ in $Q_{\tau_0,r}$. 
In the same way one proves that we get for  $\underline{I}^{\varepsilon, \delta}$ of
\begin{equation*}
\left\{
\begin{array}{rl}
\delta = & - (\underline{I}^{\varepsilon, \delta})^2  + \bigl(\underline{I}^{\varepsilon, \delta}\bigr)^2 \mathcal{G} \frac{x}{\underline{I}^{\varepsilon, \delta}} + \varepsilon \Delta \bigl(\ln{(\underline{I}^{\varepsilon, \delta})}\bigr) ; \\
\underline{I}^{\varepsilon, \delta} \big\vert_{\partial B(m,r)} = & \underline{\sigma}.
\end{array}
\right. 
\end{equation*}
subsolutions.

 {\it Third Step:} Having now super- and subsolutions, we can (by sake of Proposition \ref{pvolbounds}) invoke
 now the comparison principle Lemma \ref{volcomp} to conclude that 
\begin{equation*}
\underline{I}^{\varepsilon, \delta}(x,y) \leq I(\tau,x,y) \leq
\overline{I}^{\varepsilon, \delta}(x,y) \quad \mbox{ in } \quad
Q_{\tau_0,r} 
\end{equation*}
for all $0<\varepsilon<\varepsilon_0$,  $\varepsilon_0$ and $\tau_0$
chosen as above.  Thus we have
\begin{equation*}
\underline{I}^{\varepsilon, \delta}(x,y) \leq \liminf_{\tau \to 0}
I(\tau,x,y) \leq \limsup_{\tau \to 0} I(\tau,x,y) \leq
\overline{I}^{\varepsilon, \delta}(x,y). 
\end{equation*}
Next we want to send $\varepsilon$ to zero. Therefore we note that the
families of solutions $ \overline{I}^{\varepsilon, \delta}$,
$\underline{I}^{\varepsilon, \delta}$ are uniformly bounded (by the constants $\sqrt{\delta} \vee \overline{\sigma}$ and $\sqrt{\delta} \wedge \underline{\sigma}$ as a consequence of the comparison principle \cite[Theorem 10.7.(i)]{GilTrud} applied to the equations in the log-variables) and equicontinuous
in $\varepsilon$ (since H\"{o}lder-continuous with the same H\"{o}lder
constants). Thus by the Arzel\`{a}-Ascoli theorem $
\overline{I}^{\varepsilon, \delta}$ converges along a subsequence
uniformly on compacts to some limit function $\overline{I}^{\delta}
\in C^{0+\beta}(Q_{\tau_0,r})$. This function is a viscosity
solution of the PDE 
\begin{equation*}
\left\{
\begin{array}{rl}
-\delta =& - (\overline{I}^{\delta})^2 +\bigl(\overline{I}^{\delta}\bigr)^2 \mathcal{G} \frac{x}{\overline{I}^{\delta}}; \\
\overline{I}^{ \delta} \big\vert_{\partial B(m,r)} = &\overline{\sigma}.
\end{array}
\right. 
\end{equation*}
An analogous result holds true for the subsolutions. Now sending $\delta \to 0$, this gives by the same argument a solution of the PDE
\begin{equation*}
\mathcal{G} \frac{x}{\overline{I}^0} = 1
\end{equation*}
which satisfies $I(0,y)=0$. Thus for $\tau \to 0$, $I(\tau,x,y)$
converges locally uniformly to $I^0$ which is nothing else than the
(by \cite[Section 3.2]{BBF} unique) viscosity solution of the eikonal equation \eqref{eikonal} with
$\psi = x/I^0$.
\end{proof}

\bibliographystyle{alpha}
\bibliography{Bib_IndifferenceSmile_final.bib}
\end{document}